\documentclass[journal,twocolumn,10pt]{IEEEtran}
\usepackage[T1]{fontenc}
\usepackage{cite}
\usepackage[cmex10]{amsmath}
\allowdisplaybreaks 
\usepackage{amsthm}
\usepackage{amssymb}
\usepackage{mathtools}
\usepackage[hidelinks]{hyperref} 
\usepackage{algorithmic}
\usepackage{textcomp}
\usepackage[usenames, dvipsnames]{color}
\usepackage{algorithm,algorithmic}
\usepackage{subcaption}
\usepackage{graphicx}  
\usepackage{url}
\usepackage[usestackEOL]{stackengine}
\usepackage{verbatim}

\usepackage[overload]{empheq}
\newcommand\nbthis{\addtocounter{equation}{1}\tag{\theequation}}

\makeatletter
\newcommand{\ALOOP}[1]{\ALC@it\algorithmicloop\ #1%
	\begin{ALC@loop}}
	\newcommand{\ENDALOOP}{\end{ALC@loop}\ALC@it\algorithmicendloop}

\makeatother



\usepackage {tikz}
\tikzset{>=stealth}
\usepackage{xcolor,colortbl}
\definecolor{lightred}{RGB}{255, 240, 240}
\definecolor{lightblue}{RGB}{230, 250, 255}
\definecolor{lightgreen}{RGB}{240, 255, 242}
\definecolor{myred}{RGB}{220, 0, 0}
\definecolor{myblue}{RGB}{0, 17, 173}
\definecolor{mygreen}{RGB}{2, 117, 0}

\newtheorem{lemma}{Lemma}

\newcommand{\mca}[1]{\mathcal{#1}}

\newcommand{\mbf}[1]{\mathbf{#1}}
\newcommand{\mbb}[1]{\mathbb{#1}}
\newcommand{\mrm}[1]{\mathrm{#1}}

\newcommand{\bsm}[1]{\boldsymbol{#1}}

\newcommand{\tr}{\operatorname{tr}}

\newcommand{\fbeam}{f_{\mathrm{beam}}}

\newcommand{\fCRLB}{f_{\mathrm{CRLB}}}

\newcommand{\diag}[1]{\operatorname{diag}{#1}}

\newcommand{\vectorize}[1]{\mathrm{vec}#1}
\DeclareMathOperator*{\argmin}{arg\,min}
\DeclareMathOperator*{\argmax}{arg\,max}
\DeclareMathOperator*{\maximize}{maximize}
\DeclareMathOperator*{\minimize}{minimize}
\DeclareMathOperator*{\optimize}{optimize}
\DeclareMathOperator*{\st}{subject\;to}

\newcommand{\TT}{\mrm{T}}
\newcommand{\HH}{\mrm{H}}

\newcommand{\SINR}{\mathrm{SINR}}
\newcommand{\radar}{\mathrm{R}}
\newcommand{\RC}{\mathrm{RC}}

\newcommand{\rrm}{\mathrm{r}}

\newcommand{\lk}{{\ell}_k}
\newcommand{\lkp}{{\ell}_{k'}}
\newcommand{\CRLB}{\mrm{CRLB}}
\newcommand{\Covyr}{\bsm{\Sigma}_{\mbf{y}_{\mrm{r}}}}
\newcommand{\jj}{\mathtt{j}}

\ifCLASSINFOpdf
\else
\fi
\usepackage{amsmath}
\begin{document}
\setlength\abovedisplayskip{5.5pt}
\setlength\belowdisplayskip{5.5pt}


%
\title{Exploiting Symmetric Non-Convexity for Multi-Objective Symbol-Level DFRC Signal Design}
%

\author{
Ly~V.~Nguyen, Rang Liu, Nhan Thanh Nguyen, Markku Juntti, Bj\"{o}rn Ottersten, and A.~Lee~Swindlehurst
\thanks{Part of this paper has been presented in Proc. IEEE Int. Conf. Commun., 2024~\cite{VanLy2024Exploitation}.}
\thanks{Ly V. Nguyen, Rang Liu, and A. Lee Swindlehurst are with the Center for Pervasive Communications and Computing, Henry Samueli School of Engineering, University of California, Irvine, CA, USA 92697 (e-mail: vanln1@uci.edu, rangl2@uci.edu, swindle@uci.edu).}
\thanks{Nhan Thanh Nguyen and Markku Juntti are with the Center for Wireless Communications, University of Oulu, 90014 Oulu, Finland (email: nhan.nguyen@oulu.fi, markku.juntti@oulu.fi).}
\thanks{Bj\"{o}rn Ottersten is with the Interdisciplinary Center for Security, Reliability and Trust (SnT), University of Luxembourg, 1855 Luxembourg City, Luxembourg (e-mail: bjorn.ottersten@uni.lu)}
}

\maketitle

\begin{abstract}
Symbol-level precoding (SLP) is a promising solution for addressing the inherent interference problem in dual-functional radar-communication (DFRC) signal designs. This paper considers an SLP-DFRC signal design problem which optimizes the radar performance under communication performance constraints. We show that a common phase modulation applied to the transmit signals from an antenna array does not affect the performance of different radar sensing metrics, including beampattern similarity, signal-to-interference-plus-noise ratio (SINR), and Cram{\'e}r-Rao lower bound (CRLB). We refer to this as symmetric-rotation invariance, upon which we develop low-complexity yet efficient DFRC signal design algorithms. More specifically, we propose a symmetric non-convexity (SNC)-based DFRC algorithm that relies on the non-convexity of the radar sensing metrics to identify a set of radar-only solutions. Based on these solutions, we further exploit the symmetry property of the radar sensing metrics to efficiently design the DFRC signal. We show that the proposed SNC-based algorithm is versatile in the sense that it can be applied to the DFRC signal optimization of all three sensing metrics mentioned above (beampattern, SINR, and CRLB). In addition, since the radar sensing metrics are independent of the communication channel and data symbols, the set of radar-only solutions can be constructed offline, thereby reducing the computational complexity. We also develop an accelerated SNC-based algorithm that further reduces the complexity. Finally, we numerically demonstrate the superiority of the proposed algorithms compared to existing methods in terms of sensing and communication performance as well as computational requirements.  
\end{abstract}


%
\IEEEpeerreviewmaketitle

\section{Introduction}
\label{sec_introduction}

The explosive growth of wireless services and the proliferation of smart devices have propelled the development of integrated sensing and communication (ISAC) technologies. As a key feature for sixth-generation (6G) wireless networks, ISAC offers an integration of environmental sensing and high-speed communication, achieving significant efficiency gains through shared resources~\cite{FLiu-Tcom-2020,FLiu-JSAC-2022,JZhang-CST-2022,AKaushik-CSM-2024}. Within the ISAC framework, dual-functional radar-communication (DFRC) systems play a pivotal role, supporting simultaneous radar sensing and communication on the same platform.
This tight integration not only enhances spectral and energy efficiency but also reduces system complexity, hardware costs, and power consumption, making DFRC an attractive technology for future wireless systems. Potential applications of DFRC span autonomous driving, smart cities, environmental monitoring, and the Internet of Things (IoT), where both high communication quality and precise sensing are essential \cite{SLu-IotJ-2024}.



Attracted by its enormous potential, DFRC has been extensively studied in recent years. Dual-functional waveform design serves as a cornerstone for optimizing the performance tradeoff between radar sensing and communication functionalities \cite{JZhang-JSTSP-2021}. Among the possible approaches, spatial-domain beamforming has garnered significant attention for enhancing DFRC performance, as outlined in~\cite{FLiu-TWC-2018,XLiu-TSP-2020,ZCheng-TCCN-2021,ZCheng-JSTSP-2021,CQi-CL-2022,XWang-Tcom-2022,Liu-TSP-2022,NNguyen-JSTSP-2024,RLiu-TWC-2024}. 
These studies primarily focus on fulfilling communication requirements, such as ensuring a sufficient signal-to-interference-plus-noise ratio (SINR)\cite{FLiu-TWC-2018,XLiu-TSP-2020,CQi-CL-2022,XWang-Tcom-2022,Liu-TSP-2022} or maintaining a high data rate \cite{ZCheng-TCCN-2021,NNguyen-JSTSP-2024,RLiu-TWC-2024}, while simultaneously addressing the distinct objectives of the radar sensing tasks. These tasks typically fall into two categories: target detection, where maximizing radar SINR improves detection probability and reliability \cite{RLiu-TWC-2024}, and target parameter estimation, where minimizing the Cram{\'e}r-Rao lower bound (CRLB) enhances the accuracy of range, velocity, and angle estimation \cite{XWang-Tcom-2022,Liu-TSP-2022,RLiu-TWC-2024}. Beyond these task-specific objectives, generalized sensing performance metrics, such as beampattern design \cite{FLiu-TWC-2018,XLiu-TSP-2020,CQi-CL-2022} or waveform similarity \cite{ZCheng-TCCN-2021,NNguyen-JSTSP-2024}, have also been proposed to address the challenges of sensing in complex environments.

While these approaches have achieved notable success, from a communications perspective they attempt to minimize the impact of multi-user interference using so-called block-level processing (BLP), in which channel-coherent blocks of data undergo spatial signal processing  \cite{Spencer-CM-2004}. 
However, the time-domain properties of the signals, which are crucial for managing clutter and dynamic targets \cite{BTang-TSP-2016,LWu-TSP-2018,BTang-TSP-2020}, are not exploited.
Furthermore, DFRC systems inherently suffer from interference between radar and communication functionalities, which spatial beamforming alone cannot fully mitigate. Traditional spatial-domain beamforming methods for communication users typically treat radar signals as interference and attempt to suppress them. However, radar signals often have significantly higher power than communication signals, making complete suppression impractical and potentially degrading communication performance.

To overcome the limitations of traditional spatial beamforming, symbol-level precoding (SLP) has emerged as a transformative approach~\cite{CMasouros-TWC-2009,CMasouros-TSP-2015,HJedda-TWC-2018,MAlodeh-CST-2018,ALi-CST-2020,Ly2024CCP}. Unlike conventional BLP, which relies on second-order statistics to suppress interference, SLP directly optimizes the transmit waveform on a symbol-by-symbol basis, leveraging instantaneous multi-user symbol information to reshape interference into a constructive component at each user's receiver. This paradigm shift is particularly beneficial for DFRC, where radar-induced interference, rather than being suppressed, can be exploited to improve communication performance. Through joint optimization of radar sensing and communication functionalities, SLP-based DFRC designs harness both spatial and temporal degrees of freedom (DoFs), leading to superior system performance. Recent studies \cite{XYu-JSAC-2022,RLiu-JSTSP-2021,ZWu-JSAC-2025,BWang-WCL-2022,PHuang-GCW-2023,JZhang-GCW-2023,PLi-TWC-2024} have demonstrated that SLP-DFRC approaches significantly outperform traditional BLP-DFRC methods, achieving an improved balance between radar sensing accuracy and communication quality-of-service (QoS), making SLP a promising candidate for next-generation ISAC networks.


Early works on SLP-based DFRC primarily focused on minimizing multi-user interference (MUI) while incorporating radar performance constraints, such as waveform similarity error \cite{FLiu-TSP-2018,XYu-JSAC-2022}, transmit beampattern similarity error \cite{BTang-SAM-2020}, and the CRLB \cite{ZCheng-TWC-2021}. While these methods leverage the temporal DoFs provided by SLP to optimize the symbol-dependent waveforms, they do not fully harness the constructive interference (CI) capabilities of SLP, limiting their potential performance gains. Recognizing this deficiency, recent research has shifted towards CI-SLP-based DFRC designs, where interference from both radar and multiuser signals is no longer treated as a disruptive factor but is instead restructured to contribute constructively to symbol detection. Unlike traditional SLP methods that strictly enforce phase alignment with transmitted symbols, CI-SLP relaxes this constraint, introducing additional degrees of freedom in waveform design~\cite{ALi-CST-2020}. This increased flexibility allows for enhanced radar sensing capabilities without degrading communication quality. Various optimization objectives have been explored within this framework. Some studies have focused on optimization of beampattern similarity, ensuring that the transmitted waveform closely matches a given desired radar illumination pattern~\cite{RLiu-JSTSP-2021,ZWu-JSAC-2025,JYan-WCNC-2022,YWang-ICCC-2024}. Others have sought to maximize radar SINR, improving target detection probability~\cite{RLiu-JSAC-2022,BWang-WCL-2022}, minimize the CRLB, enhancing parameter estimation accuracy~\cite{MWang-ICCC-2022,PHuang-GCW-2023,JZhang-GCW-2023}, or minimizing the radar sidelobe levels for better range-Doppler estimation performance \cite{PLi-TWC-2024}.

Despite these advancements, the high computational complexity of existing SLP-DFRC optimization methods remains a fundamental challenge for practical deployment. The joint symbol-level optimization of radar and communication functionalities inherently leads to large-scale, non-convex, and highly nonlinear problems, posing significant difficulties for real-time implementation. Most existing approaches rely on iterative algorithmic frameworks such as alternating direction method of multipliers (ADMM), majorization maximization (MM), successive convex approximation (SCA), or the Lagrangian method with block successive upper-bound minimization (ALM-BSUM). These approaches, while effective, often suffer from slow convergence rates and substantial per-iteration computational overhead.
Moreover, many of these methods lack closed-form solutions for some variables, necessitating the use of convex optimization solvers within the iterative procedure, which further exacerbates computational complexity. Even in cases where closed-form solutions have been derived for certain subproblems~\cite{JYan-WCNC-2022,JZhang-GCW-2023,ZWu-JSAC-2025}, they often introduce auxiliary variables and additional constraints, leading to an increased number of iterations and elevated overall computational cost. 


Motivated by these challenges, in this paper we propose a computationally efficient SLP-DFRC waveform design framework that exploits the symmetric non-convexity (SNC) property of radar sensing metrics. 
While preliminary results in~\cite{VanLy2024Exploitation} demonstrated the potential of this approach for beampattern similarity optimization, this work extends the SNC-based framework to radar SINR and CRLB optimization, providing a generalized, low-complexity solution for SLP-DFRC waveform design.
Our main contributions are summarized as follows:
\begin{itemize}
    \item  We first propose a novel SLP-DFRC waveform design method that maximizes the radar beampattern similarity under symbol-level communication constraints. We observe that a common phase modulation applied to the transmit signals from an antenna array does not affect the beampattern similarity metric, a property we hereafter refer to as ``\textit{symmetric-rotation invariance}''. Motivated by this observation, we propose an SNC-based algorithm that first identifies a set of radar-only solutions leveraging the non-convexity of the radar beampattern metric. The radar-only solutions can be precomputed offline since the sensing metric is independent of the communication channel and data symbols, therefore significantly reducing computational complexity. The phase of the elements of each locally optimal radar-only waveform is then perturbed in order to (1) meet safety margin constraints for the communication users, and (2) yield a DFRC waveform that is as close as possible to the radar-only solution multiplied by a common phase modulation across all antennas.
    

    \item Next, we derive two other SLP-DFRC waveform designs that optimize the radar SINR and CRLB. We show that the symmetric-rotation invariance property also applies to the radar SINR and CRLB metrics, enabling the previously proposed SNC-based algorithm to be effectively used for DFRC signal optimization involving the SINR and CRLB. We derive closed-form expressions for the gradients of the complex SINR and CRLB metrics involved, facilitating the optimization.
    
    \item We then develop an accelerated SNC-based algorithm, hereafter referred to as the aSNC approach, that further reduces the complexity. We show analytically that the computational complexity of the proposed SNC and aSNC algorithms is significantly lower than that of other existing approaches. 

    \item Finally, we provide extensive numerical results to validate the effectiveness of the proposed algorithms, demonstrating their superior radar sensing and communication performance compared to existing approaches. Furthermore, we show significantly reduced computational complexity in multi-objective optimization scenarios.
\end{itemize}

The rest of this paper is organized as follows: Section~\ref{sec_sysytem_model_and_problem_formulation} presents the system model and formulates the problem of interest. The proposed SNC-based DFRC signal design method for radar beampattern, SINR, and CRLB metrics is presented in Section~\ref{sec:beam}, Section~\ref{sec:SINR}, and Section~\ref{sec:CRLB}, respectively. Section~\ref{sec:aR2DFRC} introduces the proposed aSNC-based method, while Section~\ref{sec:complexity} compares the computational complexity of the proposed methods with other existing approaches. Numerical results are provided in Section~\ref{sec_results}, and Section~\ref{sec_conclusion} concludes the paper.

\textit{Notation:} Upper-case and lower-case boldface letters denote matrices and column vectors, respectively. The transpose, conjugate, and conjugate transpose are denoted by $[\cdot]^{\TT}$, $[\cdot]^*$ and $[\cdot]^\HH$, respectively. The notation $\mbf{X}_{i:j,k:\ell}$ represents the sub-matrix of $\mbf{X}$ that includes rows $i$ to $j$ and columns $k$ to $\ell$. The expectation of random quantities is denoted by $\mathbb{E}[\cdot]$. The operator $\diag(\mbf{a})$ denotes a diagonal matrix whose diagonal elements are defined by the vector $\mbf{a}$. The operator $|\cdot|$ denotes the absolute value of a number. The notation $\Re\{\cdot\}$ and $\Im\{\cdot\}$ respectively denotes the real and imaginary parts of the complex argument. The notation $\mrm{Proj}(\cdot)$ is the projection function that maps its argument onto the complex unit circle, i.e., $\mrm{Proj}(x) = e^{\jj\measuredangle(x)}$ where $\measuredangle(x)$ is the angle of $x$. If $|\cdot|$, $\Re\{\cdot\}$, $\Im\{\cdot\}$, and $\mrm{Proj}(\cdot)$ are applied to a matrix or vector, they are applied separately to every element of that matrix or vector. The notation $\mathcal{CN}(\cdot,\cdot)$ represents a complex circularly-symmetric normal distribution, where the first argument is the mean and the second argument is the variance or the covariance matrix. Finally, $\jj$ is the unit imaginary number satisfying $\jj^2=-1$.

\section{System Model and Problem Formulation}
\label{sec_sysytem_model_and_problem_formulation}
\subsection{System Model}
\label{sec_system_model}
We consider a co-located monostatic multiple-input-multiple-output (MIMO) DFRC system  where a base station (BS) equipped with $N$ antennas simultaneously serves $U$ single-antenna communication users and detects the locations of $K$ targets, where it is assumed that $U \leq N$ and $K \leq N$. 
Let $\mbf{H} = [\mbf{h}_1,\,\ldots,\,\mbf{h}_{U}]^\HH \in \mbb{C}^{U \times N}$ denote the downlink channel from the BS to the communication users. The signal vector received by the users is $\mbf{y}_t = \mbf{Hx}_t+\mbf{n}_t$, where $\mbf{x}_t = [x_{1,t},\,\ldots,\,x_{N,t}]^\TT$ is the transmit DFRC signal at time slot $t = 1,\,\ldots,\,\tau$ and $\mbf{n}_{t}\sim\mca{CN}(\mbf{0},\sigma^2_{\mrm{u}}\mbf{I}_{U})$ denotes the noise at the users, where $\tau$ is the length of the transmit signal sequence. We also assume that the elements of $\mbf{x}_t$ have the same constant modulus $|x_{n,t}|^2 = P/N$ for $n = 1, \, \ldots, \, N$, where $P$ is the total transmit power. In practice, the constant-modulus constraint allows the  antennas to transmit at their maximum power to achieve the highest power efficiency, and it also guarantees low peak-to-average power ratio (PAPR), which enables the use of low-cost non-linear amplifiers. This paper focuses on the DFRC transmit signal design problem where the transmit signal $\mbf{x}_t$ is designed to simultaneously serve communication users and illuminate the targets of interest. The developed transmit signal design technique can work for either full- or half-duplex sensing systems\footnote{The impact of self-interference due to full-duplex operation on the monostatic sensing performance is an interesting future study item.}.

\subsection{Communication Performance Metric}
\label{sec_comm_metric}
Let $\mbf{s}_t \in \mca{S}^{U}$ denote the symbols intended for the users, where $\mca{S}$ is the symbol alphabet. We assume $M$-ary phase shift keying ($M$-PSK) signaling, i.e., $s_{u,t} \in \mca{S} = \exp{\big(\jj\pi\frac{2m_{u,t}+1}{M}\big)}$ where $m_{u,t} \in \{0,\,\ldots,\,M-1\}$. Let $z_{u,t} = s_{u,t}^*\mbf{h}_u^\HH\mbf{x}_t$ denote the rotated noiseless received signal of user $u$. The safety margin of $z_{u,t}$ is illustrated in Fig.~\ref{fig_safety_margine_illustration} and is defined as follows~\cite{HJedda-TWC-2018}:
\begin{align}
    \delta_{u,t} = \Re\{z_{u,t}\}  \sin(\pi/M) - |\Im\{z_{u,t}\} |\cos(\pi/M).
\end{align}
It is clear that the farther $z_{u,t}$ is from the symbol decision boundaries, the more likely that the received signal $y_{u,t}$ will be correctly detected, i.e., the more robust it will be against the effects of noise and interference. To ensure the communication QoS, the transmitted DFRC signal $\mbf{x}_t$ must be designed such that the safety margin $\delta_{u,t}$ meets or exceeds a predefined minimum threshold $\gamma_u$, i.e., $\delta_{u,t} \geq \gamma_u, \forall u, t$.

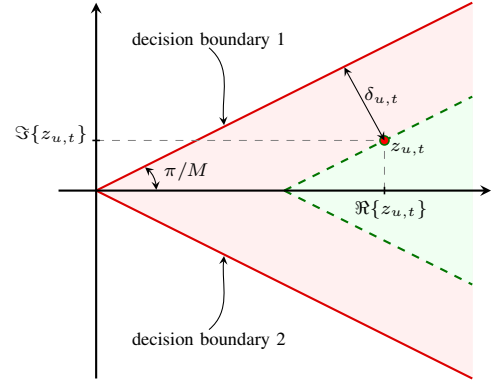
\begin{figure}
    \centering
    \begin{tikzpicture}
    \draw [lightred, fill=lightred] (0,0) -- (5,2.5) -- (5,-2.5) -- (0,0);

    \draw [lightgreen, fill=lightgreen] (2.5,0) -- (5,1.25) -- (5,-1.25) -- (2.5,0);
    
    \draw [thin, <->] (0.65,0.31) to [out=-40, in=95] (0.8,0.01);
    \node at (1.2,0.25) {{\scriptsize $\pi/M$}};
    
    \draw [thin, ->] (1.5,1.9) to [out=-40, in=95] (1.7,0.87);
    \node [rotate = 0] at (1.5,2) {{\scriptsize decision boundary 1}};
    
    \draw [thin, ->] (1.5,-1.85) to [out=40, in=-95] (1.7,-0.87);
    \node [rotate = 0] at (1.5,-2) {{\scriptsize decision boundary 2}};
    
    \draw [thick,-,myred] (0,0) to (5,2.5);
    \draw [thick,-,myred] (0,0) to (5,-2.5);
    
    \draw [thick,dashed,-,mygreen] (2.5,0) to (5,1.25);
    \draw [thick,dashed,-,mygreen] (2.5,0) to (5,-1.25);
    
    \draw [thick,->] (-0.5,0) to (5.25,0);
    \draw [thick, ->] (0,-2.5) to (0,2.5);
    
    \draw [help lines, dashed, -] (0,0.667) to (3.833,0.667);
    \draw [help lines, dashed, -] (3.833,0) to (3.833,0.667);
    \draw [mygreen, fill = red, semithick] (3.833,0.667) circle [radius=0.06];
    \node at (4.15,0.55) {{\scriptsize ${z}_{u,t}$}};
    
    \draw [thin,<->] (3.833,0.667) to (3.3,1.64);
    \node at (3.8,1.25) {{\scriptsize $\delta_{u,t}$}};
    
    \draw [semithick,-] (3.833,-0.05) to (3.833,0.05);
    \draw [semithick,-] (-0.05,0.667) to (0.05,0.667);
    \node at (3.92,-0.25) {{\scriptsize  $\Re\{z_{u,t}\}$}};
    \node at (-0.6,0.74) {{\scriptsize  $\Im\{z_{u,t}\}$}};

    
    \end{tikzpicture}
    \caption{Illustration of the safety margin $\delta_{u,t}$ of user $u$ at time slot $t$, where $z_{u,t} = s_{u,t}^*\mbf{h}_u^\HH\mbf{x}_t$ is the rotated noiseless received signal of user $u$.}
    \label{fig_safety_margine_illustration}
\end{figure}

\subsection{Radar Sensing Performance Metrics}
\label{sec_radar_metric}
We assume that the $K$ targets are located at $K$ different angles $\vartheta_1,\,\ldots,\,\vartheta_K$. When the transmit
signal $\mbf{X} = [\mbf{x}_1,\, \ldots,\, \mbf{x}_\tau]$ is reflected by the $K$ point-like targets, the radar received signal matrix is given by
\begin{equation}
    \mbf{Y}_{\rrm} = \sum_{k=1}^{K} \alpha_{k} \mbf{a}(\vartheta_k)\mbf{a}(\vartheta_k)^\HH\mbf{X} + \mbf{N}_{\rrm}\label{eq:radar_rx_signal}
\end{equation}
where $\alpha_{k} \sim \mca{CN}(0, \sigma^2_{k})$ is the amplitude that results due to path loss and the radar cross section (RCS) of target $k$, $\mbf{a}(\vartheta_k)$ is the steering vector of target $k$ at angle $\vartheta_k$, and the matrix $\mbf{N}_{\mrm{r}}$ represents noise at the radar receiver with entries that are assumed to be independent and identically distributed (i.i.d.) as $\mca{CN}(0,\sigma_{\mrm{r}}^2)$. With the radar signal model in~\eqref{eq:radar_rx_signal}, the number of transmit and receive antennas are assumed to be the same. For brevity, in the rest of the paper we will use $\mbf{a}_{k}$ in lieu of $\mbf{a}(\vartheta_k)$.

We will consider three different radar sensing performance metrics, including radar beampattern similarity, SINR, and the CRLB. In particular, the transmit DFRC signal $\mbf{X}$ will be designed to either maximize the radar beampattern similarity, maximize the radar SINR, or minimize the CRLB of the estimates of the target angles $\vartheta_1,\,\ldots,\,\vartheta_K$. We denote the sensing performance metric as $f_{\mca{A}}(\mbf{X})$, where the subscript $\mca{A}$ can represent `$\mrm{beam}$', `$\mrm{SINR}$', or `$\mrm{CRLB}$' to indicate which metric is considered. Specific expressions for these metrics will be derived in detail in subsequent sections.

\subsection{Problem Formulation}
\label{sec_problem_formulation}
We consider the problem of optimizing the radar performance under communication performance and transmit power constraints, as follows:
\begin{subequations}
\label{eq_problem1}
\begin{align*}
    \optimize_{\{\mbf{X}\}}  \quad &   f_{\mca{A}}(\mbf{X})\\
     \st \quad & \delta_{u,t} \geq \gamma_u, \; \forall u,t \nbthis \label{eq:beam_constrain1}\\
     \quad & |x_{n,t}|^2 = P/N, \; \forall n,t \nbthis \label{eq:beam_constrain2}
\end{align*}
\end{subequations}
where the operator `$\optimize$' is `$\minimize$' when the subscript $\mca{A}$ is `$\mrm{beam}$' or `$\mrm{CRLB}$', and is `$\maximize$' when the subscript $\mca{A}$ is `$\mrm{SINR}$'. Specifically, we minimize the beampattern similarity error and the CRLB of the target parameter estimates, while we maximize the received SINR of the targets. Thus, we aim to optimize the radar sensing metric, while guaranteeing that the safety margin $\delta_{u,t}$ of user $u$ at time slot $t$ is larger than or equal to a threshold~$\gamma_u$. 

As will be demonstrated later, solving problem~\eqref{eq_problem1} is challenging because the objective function $f_{\mca{A}}(\mbf{X})$ generally has a complex form and is non-convex with respect to (w.r.t.) $\mbf{X}$ for the three considered radar metrics. In the following sections, we will present the proposed SNC-based DFRC signal design methods for the beampattern, SINR, and CRLB metrics.

\section{Radar Beampattern-Based Design}
\label{sec:beam}
This section considers the beampattern similarity metric $\fbeam(\mbf{X})$, which is defined as the squared error between the actual and desired BS beampatterns. We consider a set of $L$ angles denoted by $\{\theta_1,\,\ldots,\,\theta_L\}$ where $L\gg K$ for which the beampattern similarity metric $\fbeam(\mbf{X})$ is given as~\cite{RLiu-JSTSP-2021}
\begin{align}
    \fbeam(\beta_t,\mbf{X}) = \frac{1}{L\tau}\sum_{t=1}^{\tau}\sum_{\ell=1}^L |\beta_t b_\ell - \mbf{x}^\HH_t\mbf{a}_\ell\mbf{a}_\ell^\HH\mbf{x}_t|^2,
\end{align}
where $\{\beta_t\}$ are scaling parameters, $b_\ell$ represents the desired beampattern at angle $\theta_\ell$, and $\mbf{a}_\ell$ represents the steering vector at angle $\theta_\ell$, i.e., $\mbf{a}_\ell = \mbf{a}(\theta_\ell)$, for $\ell \in \mca{L} = \{1,\,\ldots,\,L\}$. 
The problem of interest is
\begin{subequations}
\label{eq_problem1_beam}
\begin{align*}
    \minimize_{\{\beta_t,\mbf{X}\}} \quad& \fbeam(\beta_t,\mbf{X}) \nbthis \label{eq:beam_objective}\\
     \st \quad& \eqref{eq:beam_constrain1} \;\text{and}\; \eqref{eq:beam_constrain2}.
\end{align*}
\end{subequations}

Since $\fbeam(\beta_t,\mbf{X})$ is a quadratic function of $\beta_t$, the optimal $\beta_t$ that minimizes $\fbeam(\beta_t,\mbf{X})$ is
\begin{equation}
    \beta_t = \frac{\mbf{x}^\HH_t\sum_{\ell=1}^L b_\ell\mbf{a}_\ell\mbf{a}_\ell^\HH\mbf{x}_t}{\sum_{i=1}^L b_i^2}.
\end{equation}
Thus, $\fbeam(\beta_t,\mbf{X})$ can be concentrated as
\begin{equation}
    \fbeam(\mbf{X}) = \frac{1}{L\tau}\sum_{t=1}^{\tau}\sum_{\ell=1}^L |\mbf{x}^\HH_t\mbf{A}_\ell\mbf{x}_t|^2
    \label{eq_fbeam}
\end{equation}
where 
\begin{equation*}
    \mbf{A}_\ell \overset{\Delta}{=} \frac{b_\ell\sum_{i=1}^L b_{i}\mbf{a}_{i}\mbf{a}_i^\HH}{\sqrt{L}\sum_{j=1}^L b_j^2} - \frac{\mbf{a}_\ell\mbf{a}_\ell^\HH}{\sqrt{L}}.
\end{equation*}
Note that $\mbf{A}_\ell = \mbf{A}_\ell^\HH$, but it is not necessarily a positive semi-definite matrix.

The structure of $\fbeam(\mbf{X})$ in \eqref{eq_fbeam} and the separate communication constraints allow us to decompose problem \eqref{eq_problem1_beam}
into $\tau$ separate problems:
\begin{equation}
\begin{aligned}
    &\minimize_{\{\mbf{x}_t\}}  &&  \fbeam(\mbf{x_t}) = \frac{1}{L}\sum_{\ell=1}^L |\mbf{x}^\HH_t\mbf{A}_\ell\mbf{x}_t|^2\\
    & \st && \eqref{eq:beam_constrain1} \;\text{and}\; \eqref{eq:beam_constrain2}.
    \label{eq_problem1_beamt}
\end{aligned}
\end{equation}
This means that we can separately solve $\tau$ problems for the $\tau$ transmit signal vectors $\mbf{x}_1,\,\ldots,\,\mbf{x}_\tau$. 
Before addressing the solution to~\eqref{eq_problem1_beamt}, we provide the following important remarks about $\fbeam(\mbf{x}_t)$:
\begin{itemize}
    \item \textit{Remark 1:} $\fbeam(\mbf{x}_t)$ is a quartic function of $\mbf{x}_t$, which is non-convex and therefore has many local minima.
    \item \textit{Remark 2:} $\fbeam(\mbf{x}_t)$ does not depend on the communication channel $\mbf{H}$ and the users' data symbols $\mbf{s}$.
    \item \textit{Remark 3:}  The transmit signal $\mbf{x}_t$ can undergo an arbitrary phase modulation without impacting the radar performance, i.e., $\fbeam(\mbf{x}_t) = \fbeam(e^{\jj\varphi_t}\mbf{x}_t)$ $\forall \varphi_t$, a fact that will be exploited later in Section~\ref{sec:DFRC_solution} for efficient DFRC signal design.
\end{itemize}

Based on these properties, we will approach finding a solution to~\eqref{eq_problem1_beamt} using a two-step procedure. First, we find a set of $D$ locally optimal ``radar-only'' solutions to~\eqref{eq_problem1_beamt} by ignoring the communication constraints~\eqref{eq:beam_constrain1} and~\eqref{eq:beam_constrain2}. Such a set of solutions is easily found due to Remark~1 and is independent of the time index $t$ due to Remark~2, and thus can be calculated offline. We refer to this initial set of offline solutions as $\mca{X}^{\radar} = \{\mbf{x}_1^{\radar},\,\ldots,\,\mbf{x}_D^{\radar}\}$. Then, for each radar-only solution $\mbf{x}_d^{\radar} \in \mca{X}^{\radar}$, we exploit Remark~3 to efficiently find a DFRC solution $\mbf{x}_{d,t}^{\RC}$ that satisfies the communication constraints. Finally, among the $D$ DFRC solutions $\{\mbf{x}_{1,t}^{\RC},\,\ldots,\,\mbf{x}_{D,t}^{\RC}\}$, we choose the one that gives the best radar performance as the DFRC transmit signal $\mbf{x}_t^{\RC}$ at time $t$. We use the superscripts `$\mrm{R}$' and `$\mrm{RC}$' to indicate `radar-only' and `radar-communication', respectively. 

Details of our proposed SNC-based solution are presented below. In this and subsequent discussions in later sections for different radar performance metrics, we will drop the subscript $t$ indicating the symbol time index to simplify the notation.

\subsection{Construction of $\mca{X}^{\radar}$}
\label{sec:radar_solutions}
The problem of interest is
\begin{equation}
\begin{aligned}
    &\minimize_{\{\mbf{x}\}}  &&   \sum_{\ell=1}^L |\mbf{x}^\HH\mbf{A}_\ell\mbf{x}|^2\\
    & \st && \eqref{eq:beam_constrain2},
    \label{eq_problem2}
\end{aligned}
\end{equation}
which is non-convex and thus has many local optima, as mentioned earlier in Remark 1. Non-convex optimization problems with many local stationary points are generally considered to be challenging. However, in this paper, we exploit the non-convex structure to construct the set $\mca{X}^{\radar}$. First we use a set of random initializations to find a set of $\tilde{D}$ stationary points $\mca{\Tilde{X}}^{\radar}$, where $\tilde{D} \gg D$. Then, from $\mca{\Tilde{X}}^{\radar}$, we choose the $D$ solutions that give the best radar performance to obtain the set $\mca{X}^\radar$. The $\tilde{D}$ random initialization points are updated by moving along the opposite direction of the Riemannian gradient until convergence. This can be easily achieved since the gradient-based method guarantees convergence to a local optimum. The Riemannian gradient of a function $f(\mbf{x})$ can be obtained by projecting the Euclidean gradient $\nabla f(\mbf{x})$ onto the tangent space of the complex unit circle as follows: 
$$
\nabla_{\mrm{Rie}} f(\mbf{x}) = \nabla f(\mbf{x})  - \Re\left\{\nabla f(\mbf{x}) \odot\mbf{x}^*\right\}\odot\mbf{x},
$$
where $\odot$ denotes the element-wise multiplication operator, and $\nabla f(\mbf{x})$ is computed as 
\begin{equation*}
    \nabla f(\mbf{x}) = 2\sum_{\ell=1}^L \mbf{x}^\HH\mbf{A}_\ell\mbf{x}\frac{\partial \big(\mbf{x}^\HH\mbf{A}_\ell\mbf{x}\big)}{\partial\mbf{x}^\HH} = 2\sum_{\ell=1}^L \mbf{x}^\HH\mbf{A}_\ell\mbf{x}\mbf{A}_\ell\mbf{x}.
\end{equation*}
The search for a local radar optimum is implemented in the following iterative manner:
$$
\mbf{x} \leftarrow \sqrt{P/N}\, \mrm{Proj} \left(\mbf{x} - \eta_1\nabla_{\mrm{Rie}} f\big(\mbf{x}\big)\right) \; ,
$$
where $\eta_1$ is the step size. 

Since the local radar solutions in $\mca{\Tilde{X}}^{\radar}$ are found with random initializations, some may be identical or they may only differ by a complex rotation $e^{\jj\varphi}$ as indicated in Remark 3. Thus, we eliminate these duplicate solutions from $\mca{\Tilde{X}}^{\radar}$ before choosing the $D$ best solutions to form the set $\mca{X}^\radar$. In particular, we treat any pair $\mbf{x}_{\Tilde{d}}$ and $\mbf{x}_{\Tilde{d}'}$ with $\Tilde{d} \neq \Tilde{d}'$ as duplicates if $\operatorname{Var}(\diag(\mbf{x}_{\Tilde{d}})^{-1}\mbf{x}_{\Tilde{d}'})\leq \epsilon$ for some small threshold $\epsilon$, where $\operatorname{Var}(\mbf{x})$ defines the sample variance of the elements in vector~$\mbf{x}$. 

\subsection{Design of DFRC Signal $\mbf{x}^{\RC}$}
\label{sec:DFRC_solution}
Since the DFRC signal ${x}^{\RC}_n$ and the radar signal ${x}^{\radar}_n$ both have the same constant modulus, we can write
\begin{equation}\label{eq:RCsol}
    \mbf{x}^{\RC} = \diag(\mbf{x}^{\radar})\bsm{{\phi}},
\end{equation}
where $\phi_n = e^{\jj\varphi_n}, \varphi_n \in [0, 2\pi], \forall n$. Thus, given a radar solution $\mbf{x}^{\radar}$, the design of $\mbf{x}^{\RC}$ is equivalent to the design of the vector of rotations $\bsm{{\phi}}$, which are found such that $\mbf{x}^{\RC}$ satisfies the communication constraints~\eqref{eq:beam_constrain1} and~\eqref{eq:beam_constrain2}. To avoid adversely affecting the radar performance, we decompose $\bsm{{\phi}}$ into two components as follows:
\begin{equation}
    \bsm{{\phi}} = e^{\jj\varphi}\bsm{{\tilde{\phi}}},
    \label{eq:rotation_decompose}
\end{equation}
where $\bsm{{\tilde{\phi}}}$ is desired to be as close as possible to $\mbf{1}_N$, the $N$-dimensional vector whose elements all equal one. If $\bsm{{\tilde{\phi}}}=\mbf{1}_N$, then the DFRC signal in \eqref{eq:RCsol} will have the same radar performance as $\mbf{x}^{\radar}$ due to Remark~3. Thus, the basic idea of our algorithm presented next is to design $\varphi$ and a vector $\bsm{{\tilde{\phi}}}$ close to $\mbf{1}_N$ to ensure the communication constraints. The steps of the algorithm are outlined below.

\subsubsection{Optimize $\varphi$ for a given $\bsm{\Tilde{\phi}}$}
\label{sec_design_varphi}
Using the decomposition in~\eqref{eq:rotation_decompose}, the DFRC signal $\mbf{x}^\RC$ can be written as
$$
\mbf{x}^\RC = e^{\jj\varphi}\diag(\mbf{x}^{\radar})\bsm{{\tilde{\phi}}}.
$$
The rotated noiseless received signal is then given as
$$
\mbf{z} = \diag(\mbf{s}^*)\mbf{H}\mbf{x}^\RC = e^{\jj\varphi}\mbf{Q}\bsm{{\tilde{\phi}}} = e^{\jj{\varphi}}\mbf{\tilde{s}},
$$
where $\mbf{Q} = \diag(\mbf{s}^*)\mbf{H}\diag(\mbf{x}^{\radar})$ and $\mbf{\tilde{s}} = \mbf{Q}\bsm{\tilde{\phi}}$. 

Let $\mbf{\tilde{s}} = [\tilde{s}_1,\,\ldots,\,\tilde{s}_{U}]^\TT$. The safety margin of user $u$ can be written as a function of ${\varphi}$ as follows: 
\begin{align}
    \delta_u(\varphi) &= \Re\{e^{\jj\varphi}\tilde{s}_u\}  \sin(\pi/M) - |\Im\{e^{\jj\varphi}\tilde{s}_u\} |\cos(\pi/M).
\end{align}
To ensure the constraints in~\eqref{eq:beam_constrain1}, our objective will be to find an angle $\varphi$ such that $\min_{u}\, \big(\delta_u(\varphi) - \gamma_u\big)$ is maximized. If $\varphi$ can make $\min_{u}\, \big(\delta_u(\varphi) - \gamma_u\big) \geq 0$, then the communication constraints are satisfied. Hence, we need to solve the following optimization problem:
\begin{equation}
\begin{aligned}
    \maximize_{0 \leq\varphi  \leq 2\pi}  \;\;  \min_{u}\, \big(\delta_u(\varphi) - \gamma_u\big).
\end{aligned}
\label{eq:varphi_optimize}
\end{equation}
This is a non-convex problem but it can be solved easily since it involves only the single variable $\varphi$ which is constrained to lie in $[0, 2\pi]$. We propose an efficient algorithm to solve~\eqref{eq:varphi_optimize} as follows.

First, we uniformly sample the range $[0, 2\pi]$ to obtain $C$ angle samples $\{\check{\varphi}_1,\,\ldots,\,\check{\varphi}_C\}$, i.e., $\check{\varphi}_c = 2\pi(c-1)/(C-1)$. Then, a coarse solution for $\varphi$ can be found as the value in $\{\check{\varphi}_1,\,\ldots,\,\check{\varphi}_C\}$ that yields the largest $\min_{u}\, \big(\delta_u(\check{\varphi_c}) - \gamma_u\big)$, i.e.,
\begin{equation}
    \varphi^{\mrm{coarse}} = \argmax_{\{\check{\varphi}_c\}} \; \min_{u}\, \big(\delta_u(\check{\varphi_c}) - \gamma_u\big).
    \label{eq:varphi_coarse}
\end{equation}
Next, starting from the coarse solution $\varphi^{\mrm{coarse}}$, a refined solution  $\varphi^{\mrm{fine}}$ can be obtained by moving along the gradient direction until convergence:
\begin{equation}
    \varphi^{\mrm{fine}} \leftarrow \varphi^{\mrm{fine}} + \eta_2\frac{\partial \delta_{\bar{u}}(\varphi^{\mrm{fine}})}{\partial \varphi},
    \label{eq:varphi_update}
\end{equation}
where $\eta_2$ is the step size, $\Bar{u}$ is the user index for which $\Bar{u} = \argmin_{u} \big(\delta_u(\varphi^{\mrm{fine}}) - \gamma_u\big)$, and the gradient is given as
\begin{align}
    \frac{\partial \delta_{\Bar{u}}(\varphi)}{\partial \varphi}
    & = (-\sin(\varphi)\Re\{\tilde{s}_{\Bar{u}}\} - \cos(\varphi)\Im\{\Tilde{s}_{\Bar{u}}\})\sin(\pi/M) \;- \notag \\
    &\qquad \frac{\cos(\varphi)\Im\{\tilde{s}_{\Bar{u}}\} + \sin(\varphi)\Re\{\Tilde{s}_{\Bar{u}}\}}{|\cos(\varphi)\Im\{\tilde{s}_{\Bar{u}}\} + \sin(\varphi)\Re\{\Tilde{s}_{\Bar{u}}\}|}\;\times \notag \\&\qquad (-\sin(\varphi)\Im\{\tilde{s}_{\Bar{u}}\} + \cos(\varphi)\Re\{\Tilde{s}_{\Bar{u}}\})\cos(\pi/M).\notag
\end{align}
We choose to find $\varphi_{\mrm{fine}}$ from the coarse solution in~\eqref{eq:varphi_coarse} instead of directly from a random sample in $[0, 2\pi]$ because the objective function in~\eqref{eq:varphi_optimize} is non-convex. Directly moving along the gradient direction from a random sample in $[0, 2\pi]$ will likely lead to a local solution.

\subsubsection{Update $\bsm{\tilde{\phi}}$ for a given $\varphi$}
Let $\mbf{Q} = [\mbf{q}_1,\,\ldots,\,\mbf{q}_U]^\HH$ and
\begin{align*}
\mbf{\tilde{q}}_{2u} &= e^{-\jj{\varphi}}\mbf{q}_u\big[\sin(\pi/M)+e^{\jj\pi/2}\cos(\pi/M)\big],\\
\mbf{\tilde{q}}_{2u-1} &= e^{-\jj{\varphi}}\mbf{q}_u\big[\sin(\pi/M)-e^{\jj\pi/2}\cos(\pi/M)\big].
\end{align*}
Then the constraint $\delta_u \geq \gamma_u$ is equivalent to the following two conditions:
\begin{align}
    &\Re\big\{\mbf{\tilde{q}}_{2u}^\HH\bsm{\Tilde{\phi}}\big\} \geq \gamma_{u}\label{eq:cond1}\\
    &\Re\big\{\mbf{\tilde{q}}_{2u-1}^\HH\bsm{\Tilde{\phi}}\big\} \geq \gamma_{u}.\label{eq:cond2}
\end{align}
Here, $\Re\big\{\mbf{\tilde{q}}_{2u}^\HH\bsm{\Tilde{\phi}}\big\}$ and $\Re\big\{\mbf{\tilde{q}}_{2u-1}^\HH\bsm{\Tilde{\phi}}\big\}$ are the margins to the decision boundaries~1 and~2 in Fig.~\ref{fig_safety_margine_illustration}, respectively.

Consider a function of $\bsm{\Tilde{\phi}}$ as follows:
\begin{equation}
    \xi_{u'}(\bsm{\Tilde{\phi}}) = \Re\big\{\mbf{\tilde{q}}_{u'}^\HH\bsm{\Tilde{\phi}}\big\} - \Tilde{\gamma}_{u'},
\end{equation}
where $u' = 1,\,\ldots,\,2U$ and $\Tilde{\gamma}_{u'} = \gamma_u$ if $u' = 2u$ or $u' = 2u-1$, respectively. We need to achieve $\min_{u'} \xi_{u'}(\bsm{\Tilde{\phi}}) \geq 0$ because the two conditions~\eqref{eq:cond1} and~\eqref{eq:cond2} are satisfied for all $u$ when $\min_{u'} \xi_{u'}(\bsm{\Tilde{\phi}}) \geq 0$. Therefore, we update $\bsm{\Tilde{\phi}}$ by moving along the Riemannian gradient of the function $\xi_{\bar{u}}(\bsm{\Tilde{\phi}})$, where 
$$
\bar{u} = \argmin_{u' = 1,\, \ldots,\, 2U} \xi_{{u'}}(\bsm{\Tilde{\phi}}).
$$
The Riemannian gradient of $\xi_{\Bar{u}}(\bsm{\Tilde{\phi}})$ is given as
$$
\nabla_{\mrm{Rie}}\;\xi_{\Bar{u}}(\bsm{\Tilde{\phi}}) = \mbf{\tilde{q}}_{\bar{u}} - \Re\left\{\mbf{\tilde{q}}_{\bar{u}}\odot\bsm{\Tilde{\phi}}^*\right\}\odot\bsm{\Tilde{\phi}},
$$
and thus an update of $\bsm{\Tilde{\phi}}$ can be obtained as follows:
\begin{equation}
    \bsm{\Tilde{\phi}} \leftarrow \mrm{Proj} \left(\bsm{\Tilde{\phi}} + \eta_3\nabla_{\mrm{Rie}}\;\xi_{\Bar{k}}\big(\bsm{\Tilde{\phi}}\big)\right),
    \label{eq:phi_update}
\end{equation}
where $\eta_3$ is a given step size.

It is important to note that, for a given $\varphi$, the update of $\bsm{\Tilde{\phi}}$ in~\eqref{eq:phi_update} is implemented only once as described in Algorithm~\ref{algo_R2DFRC}, since a change in $\bsm{\Tilde{\phi}}$ affects the radar performance while a change in $\varphi$ does not. In addition, $\bsm{\Tilde{\phi}}$ is desired to be as close to $\mbf{1}_N$ as possible to avoid radar performance loss. Therefore, we should pursue a minimal number of updates on $\bsm{\Tilde{\phi}}$ and maximally exploit the flexibility provided by $\varphi$. Updates for $\bsm{\Tilde{\phi}}$ and $\varphi$ are terminated as soon as the communication constraints are all satisfied. The final DFRC signal design $\mbf{x}^\RC$ from Algorithm~\ref{algo_R2DFRC} is the one in $\{\mbf{x}_1^{\RC},\,\ldots,\,\mbf{x}_D^{\RC}\}$ that gives the best radar performance, i.e., $\mbf{x}^\RC = \argmin_{\{\mbf{x}^\RC_d\}}\; f(\mbf{x}^\RC_d)$.

\begin{algorithm}[t!]
\caption{Proposed SNC-based DFRC Signal Design Algorithm.}
\label{algo_R2DFRC}
    \begin{algorithmic}[1]
        \REQUIRE $\mca{X}^{\radar}$, $\mbf{H}$, and $\mbf{s}$
        \ENSURE $\mbf{x}^{\RC}$
        \FOR{$d = 1,\, \ldots,\, D$}
        \STATE Initialize $\bsm{\Tilde{\phi}}_d = \mbf{1}_N$\label{algo1_line2}
        \STATE Compute $\mathbf{Q}_d = \diag(\mbf{s}^*)\mbf{H}\diag(\mbf{x}^{\radar}_d)$\label{algo1_line3}
        \WHILE{constraints $\delta_u \geq \gamma_u$ are not satisfied $\forall u$}
                \STATE Compute $\mbf{\tilde{s}}_d = \mbf{Q}_d\bsm{\tilde{\phi}}_d$\label{algo1_line5}
                \STATE Find $\varphi_d^{\mrm{coarse}}$ by~\eqref{eq:varphi_coarse} then set $\varphi_{d}^{\mrm{fine}} = \varphi_{d}^{\mrm{coarse}}$\label{algo1_line6}
                \STATE Update $\varphi^\mrm{fine}_d$ by~\eqref{eq:varphi_update} until convergence\label{algo1_line7}
                \STATE Set $\varphi_d = \varphi_{d}^{\mrm{fine}}$

                \STATE If $\delta_u \geq \gamma_u$ $\forall u$, then
                \textbf{exit} the while loop\label{algo1_line9}
                
                \STATE Update $\bsm{\Tilde{\phi}}_d$ by~\eqref{eq:phi_update}\label{algo1_line10}
    
                \STATE If $\delta_u \geq \gamma_u$ $\forall u$, then
                \textbf{exit} the while loop\label{algo1_line11}
                
        \ENDWHILE
        \STATE Set $\mbf{x}^\RC_d = e^{\jj\varphi_d}\diag(\mbf{x}^{\radar}_d)\bsm{{\tilde{\phi}}}_d$
        \ENDFOR
        \RETURN $\mbf{x}^\RC = \argmin_{\{\mbf{x}^\RC_d\}}\; f(\mbf{x}^\RC_d)$
    \end{algorithmic}
\end{algorithm}

\section{Radar SINR-Based Design}
\label{sec:SINR}
In the previous section, we showed that Remarks~1-3 hold for the beampattern metric, based on which we developed an efficient SNC-based transmit DFRC signal design. In this section, we show that Remarks~1-3 also hold for the radar SINR metric. Thus, the SNC-based method proposed in the previous section can be adapted to solve the problem of optimizing the radar SINR under communication constraints.

To derive the SINR for target $k$, we vectorize the received signal matrix $\mbf{Y}_{\rrm}$ to obtain
\begin{align}
    \mbf{y}_{\rrm} &= \sum_{k=1}^{K} \alpha_{k} \big(\mbf{I}_\tau\otimes(\mbf{a}_{k}\mbf{a}_{k}^\HH)\big)\mbf{x} + \mbf{n}_{\rrm} =\sum_{k=1}^{K} \alpha_{k} \mbf{B}_{k}\mbf{x} + \mbf{n}_{\rrm} ,
\end{align}
where $\mbf{y}_{\rrm} = \vectorize{(\mbf{Y}_{\rrm})}$, $\mbf{B}_{k} = \mbf{I}_\tau\otimes(\mbf{a}_{k}\mbf{a}_{k}^\HH)$, $\mbf{x} = \vectorize{(\mbf{X})}$, and $\mbf{n}_{\rrm} = \vectorize{(\mbf{N}_{\rrm})}$. The SINR for target $k$ is thus 
\begin{equation}
    f_{\SINR_k}(\mbf{x},\mbf{w}_{k}) = \displaystyle\frac{\sigma_{k}^2|\mbf{w}_{k}^\HH\mbf{B}_{k}\mbf{x}|^2}{\mbf{w}_{k}^\HH\left(\sum_{i\neq \ell_k} \sigma^2_i\mbf{B}_{i}\mbf{x}\mbf{x}^\HH\mbf{B}_{i}^\HH + \sigma^2_\rrm\mbf{I}\right)\mbf{w}_{k}},
    \label{eq:SINR_kt}
\end{equation}
where $\ell_k \in \mca{L}$ is the angle index of target $k$ in $\mca{L}$ and $\mbf{w}_k$ is the radar receive filter for target $k$. The location indices corresponding to other targets and clutter are in the set $\mca{L} \setminus \{\ell_k\}$. 

We want to design the transmit signal $\mbf{x}$ and the receive radar filters $[\mbf{w}_1,\, \ldots,\, \mbf{w}_{K}]$ so that the minimum radar SINR is maximized. Hence, the problem of interest is given as
\begin{equation}
\begin{aligned}
    &\maximize_{\{\mbf{x},\,\mbf{w}_1,\,\ldots,\,\mbf{w}_{K}\}}  &&  \min_{k}\,f_{\SINR_k}(\mbf{x},\mbf{w}_{k})\\
    & \st && \eqref{eq:beam_constrain1} \;\text{and}\; \eqref{eq:beam_constrain2}.
    \label{eq:problem2}
\end{aligned}
\end{equation}
For a given transmit signal vector $\mbf{x}$, problem~\eqref{eq:problem2} reduces to a minimum variance distortionless response optimization problem~\cite{BTang-TSP-2020,LWu-TSP-2018}
\begin{equation}
\begin{aligned}
    &\minimize_{\{\mbf{w}_{k}\}}  &&  \mbf{w}_{k}^\HH\bigg(\sum_{i\neq \ell_k} \sigma^2_i\mbf{B}_{i}\mbf{x}\mbf{x}^\HH\mbf{B}_{i}^\HH + \sigma^2_\rrm\mbf{I}\bigg)\mbf{w}_{k}\\
    & \st && \mbf{w}_{k}^\HH\mbf{B}_{k}\mbf{x} = 1,
    \label{eq:problem2_wkt}
\end{aligned}
\end{equation}
whose optimal solution $\mbf{w}_{k}$ is given as
\begin{equation}
    \mbf{w}_{k} = \frac{\big(\sum_{i\neq \ell_k} \sigma^2_{i}\mbf{B}_{i}\mbf{x}\mbf{x}^\HH\mbf{B}_{i}^\HH + \sigma^2_\rrm\mbf{I}\big)^{-1}\mbf{B}_{k}\mbf{x}}{\mbf{x}^\HH\mbf{B}_{k}^\HH\big(\sum_{i\neq \ell_k} \sigma^2_i\mbf{B}_{i}\mbf{x}\mbf{x}^\HH\mbf{B}_{i}^\HH + \sigma^2_\rrm\mbf{I}\big)^{-1}\mbf{B}_{k}\mbf{x}}.
    \label{eq:optimal_wkt}
\end{equation}
Substituting~\eqref{eq:optimal_wkt} into~\eqref{eq:SINR_kt} yields the SINR for target $k$ that only depends on the transmit signal $\mbf{x}$ as follows:
\begin{equation}
    f_{\SINR_k} (\mbf{x}) = \sigma^2_{k}\mbf{x}^\HH\mbf{B}_{k}^\HH\mbf{C}_{k}^{-1}\mbf{B}_{k}\mbf{x}\label{eq:SINRk},
\end{equation}
where
\begin{equation}
    \mbf{C}_{k} = \sum_{i\neq \ell_k} \sigma^2_i\mbf{B}_{i}\mbf{x}\mbf{x}^\HH\mbf{B}_{i}^\HH + \sigma^2_\rrm\mbf{I}.\label{eq:Ck}
\end{equation}
Thus, problem~\eqref{eq:problem2} can be written in the following form:
\begin{equation}
\begin{aligned}
    &\maximize_{\{\mbf{x}\}}  &&  \min_{k}\, f_{\SINR_k} (\mbf{x})\\
    & \st && \eqref{eq:beam_constrain1} \;\text{and}\; \eqref{eq:beam_constrain2},
    \label{eq:problem2_x}
\end{aligned}
\end{equation}
which now only depends on the transmit signal $\mbf{x}$.

The function $f_{\SINR_k} (\mbf{x})$ in~\eqref{eq:SINRk} has a complicated form that involves a matrix inversion, making it challenging to solve problem~\eqref{eq:problem2_x}. However, we will exploit the fact that Remarks~1-3 continue to hold for the SINR function $f_{\SINR_k} (\mbf{x})$, as described below:
\begin{itemize}
    \item $f_{\SINR_k}(\mbf{x})$ is non-concave w.r.t.\ $\mbf{x}$ and therefore has many local maxima.
    \item $f_{\SINR_k}(\mbf{x})$ does not depend on the communication channel $\mbf{H}$ and the users' data symbols $\mbf{s}$.
    \item If each signal vector $\mbf{x}_t$ is scaled by $e^{\jj\varphi_t}$ for any $\varphi_t$ and we let $\bsm{\Omega} = \diag(e^{\jj\varphi_1},\,\ldots,\,e^{\jj\varphi_\tau})$, then the rotated transmit signal $\mbf{\Tilde{x}} = \big(\bsm{\Omega}\otimes \mbf{I}_N\big) \mbf{{x}}$ does not change the radar SINR, i.e.,
    \begin{equation}
        f_{\SINR_k}(\mbf{\Tilde{x}}) = f_{\SINR_k}(\mbf{x}),\ \forall \varphi_1,\,\ldots,\,\varphi_\tau. \label{eq:SINR_invariance}
    \end{equation}
    The result in~\eqref{eq:SINR_invariance} can be easily obtained by calculating the SINR after right-multiplying the received signal $\mbf{Y}_{\mrm{r}}$ in~\eqref{eq:radar_rx_signal} with the rotation matrix $\bsm{\Omega}$.
\end{itemize}
Hence, we can directly adopt the SNC-based method proposed in the previous section to solve problem~\eqref{eq:problem2_x}. The main difference is in how the set of prior radar solutions is generated, as this requires calculation of the gradient of $f_{\SINR_k} (\mbf{x})$ w.r.t. $\mbf{x}$, as derived in the following lemma.
\begin{lemma}
    \label{lemma1}
    The Euclidean gradient of $f_{\SINR_k} (\mbf{x})$ w.r.t. $\mbf{x}$ is given as
    \begin{align}
        &\nabla f_{\SINR_k} (\mbf{x}) = \notag \\
        &\quad\sigma^2_{k}\Big(\mbf{B}_{k}^\HH - \sum_{i\neq \ell_k}\sigma^2_i(\mbf{x}^\HH\mbf{B}_{k}^\HH\mbf{C}_k^{-1}\mbf{B}_i\mbf{x})\mbf{B}_i^\HH\Big)\mbf{C}_k^{-1}\mbf{B}_{k}\mbf{x}.
        \label{eq:grad_eta_k}
    \end{align}
\end{lemma}
\begin{proof}
    See Appendix~\ref{appendix_1}.
\end{proof}
Since SINR-based objective is $\min_{k}\, f_{\SINR_k} (\mbf{x})$, the gradient for a given $\mbf{x}$ has to be taken for the target with the lowest SINR. To design a transmit DFRC signal vector $\mbf{x}_t$, we can use Algorithm~\ref{algo_R2DFRC} since, as noted in~\eqref{eq:SINR_invariance}, an arbitrary rotation $e^{\jj\varphi_t}$ to the transmit signal vector $\mbf{x}_t$ does not change the radar SINR.



\section{Radar CRLB-Based Design}
\label{sec:CRLB}
Thus far, we have exploited the SNC structure of both the radar beampattern and SINR mectrics to develop efficient DFRC signal designs. In this section, we demonstrate that the SNC property extends to the radar CRLB metric, enabling the proposed SNC-based method to be directly applied for CRLB-based designs.

First, let us denote $\bsm{\psi} = [\alpha_{1},\,\ldots,\,\alpha_{K},\,\vartheta_1,\,\ldots,\,\vartheta_K]^\TT$ as the vector of target parameters and let $\mbf{G} = \sum_{k=1}^{K} \alpha_{k} \mbf{a}_{k}\mbf{a}_{k}^\HH$, then we can write the received signal $\mbf{Y}_{\rrm}$ as follows:
\begin{align}
    &\mbf{Y}_{\rrm} = \mbf{G}\mbf{X} + \mbf{N}_{\rrm}.
\end{align}
We also vectorize the received signal $\mbf{Y}_{\rrm}$ to obtain
\begin{align}
    \mbf{y}_{\rrm} = (\mbf{X}^\TT \otimes \mbf{I}_N)\mbf{g} + \mbf{n}_{\rrm},
\end{align}
where $\mbf{g} = \vectorize{(\mbf{G})}$.


Given $\mbf{X}$ and $\bsm{\psi}$, the received signal $\mbf{y}_{\rrm}$ is distributed as $\mca{CN}(\bsm{\mu}, \sigma_{\mrm{r}}^2\mbf{I}_N)$ where $\bsm{\mu} = (\mbf{X}^\TT \otimes \mbf{I}_N)\mbf{g}$. The Fisher information matrix w.r.t. $\bsm{\psi}$ is given by~\cite{Pakrooh2015Analysis}
\begin{align}
    \mbf{F} &=  \left(\frac{\partial \bsm{\mu}}{\partial \bsm{\psi}}\right)^\HH (\sigma_{\mrm{r}}^2\mbf{I}_N)^{-1}\frac{\partial \bsm{\mu}}{\partial \bsm{\psi}} \notag \\
    &=\frac{1}{\sigma^2_{\mrm{r}}}\left(\frac{\partial \mbf{g}}{\partial \bsm{\psi}}\right)^\HH \big((\mbf{X}^*\mbf{X}^\TT)\otimes\mbf{I}_N\big) \frac{\partial \mbf{g}}{\partial \bsm{\psi}} \notag \\
    &= \frac{1}{\sigma^2_{\mrm{r}}}\mbf{T}^\HH \big((\mbf{X}^*\mbf{X}^\TT)\otimes\mbf{I}_N\big) \mbf{T} = \frac{1}{\sigma^2_{\mrm{r}}}\mbf{T}^\HH \mbf{\tilde{X}} \mbf{T},
\end{align}
where we have defined $\mbf{T} =
\frac{\partial \mbf{g}}{\partial \bsm{\psi}}$ and $\mbf{\tilde{X}} = (\mbf{X}^*\mbf{X}^\TT)\otimes\mbf{I}_N$.

While the result we describe holds for any array geometry, assume as an example that a uniform linear array (ULA) with half-wavelength antenna spacing is employed, so that the steering vector is expressed as 
$$\mbf{a}_{k} = \frac{1}{\sqrt{N}}[1,\,e^{\jj\pi\sin(\vartheta_{k})},\,\ldots,\,e^{\jj(N-1)\pi\sin(\vartheta_{k})}]^\HH.$$
Define
\begin{align*}
    \mbf{\dot{a}}_{k} = \frac{\partial \mbf{a}_{k}}{\partial \vartheta_{k}} &=\frac{1}{\sqrt{N}}[0,\,\jj\pi\cos(\vartheta_{k})e^{\jj\pi\sin(\vartheta_{k})},\,\ldots,\,\\
    &\qquad  \jj\pi(N-1)\cos(\vartheta_{k})e^{\jj\pi(N-1)\sin(\vartheta_{k})}]^\HH,
\end{align*}
and note that $\mbf{T}$ can then be written in the following form
\begin{equation}
    \mbf{T} = [\mbf{T}_\alpha \; \, \mbf{T}_\vartheta]
\end{equation}
where $\mbf{T}_\alpha$ and $\mbf{T}_\vartheta$ are given in~\eqref{eq:T_alpha} and~\eqref{eq:T_theta}, respectively. The Fisher information matrix $\mbf{F}$ can now be expressed as
\begin{equation}
    \mbf{F} = \frac{1}{\sigma^2_{\mrm{r}}}\begin{bmatrix}
        \mbf{T}_\alpha^\HH\mbf{\tilde{X}}\mbf{T}_\alpha & \mbf{T}_\alpha^\HH\mbf{\tilde{X}}\mbf{T}_\vartheta \\
        \mbf{T}_\vartheta^\HH\mbf{\tilde{X}}\mbf{T}_\alpha & \mbf{T}_\vartheta^\HH\mbf{\tilde{X}}\mbf{T}_\vartheta
    \end{bmatrix}.
\end{equation}
\begin{figure*}[t!]
    \begin{equation}
    \mbf{T}_\alpha = \left[ \begin{array}{ccc}
        \mbf{a}_{1} & & \mbf{a}_{K}\\
        e^{-\jj\pi\sin(\vartheta_{1})}\mbf{a}_{1} & & e^{-\jj\pi\sin(\vartheta_{K})}\mbf{a}_{K}\\
        \vdots & \cdots & \vdots \\
        e^{-\jj\pi(N-1)\sin(\vartheta_{1})}\mbf{a}_{1} & & e^{-\jj\pi(N-1)\sin(\vartheta_{K})}\mbf{a}_{K} \end{array}
    \right] \label{eq:T_alpha}
    \end{equation}
    \begin{equation}
    \mbf{T}_\vartheta = \left[ \begin{array}{ccc}
        \alpha_{1}\mbf{\dot{a}}_{1} & & \alpha_{K}\mbf{\dot{a}}_{K}\\
        \alpha_{1}e^{-\jj\pi\sin(\vartheta_{1})}(\mbf{\dot{a}}_{1}-\jj\pi\cos(\vartheta_{1})\mbf{a}_{1}) & & \alpha_{K}e^{-\jj\pi\sin(\vartheta_{K})}(\mbf{\dot{a}}_{K}-\jj\pi\cos(\vartheta_{K})\mbf{a}_{K})\\
        \vdots & \cdots & \vdots \\
        \alpha_{1}e^{-\jj\pi (N-1)\sin(\vartheta_{1})}(\mbf{\dot{a}}_{1}-\jj\pi(N-1)\cos(\vartheta_{1})\mbf{a}_{1}) & & \alpha_{K}e^{-\jj\pi (N-1)\sin(\vartheta_{K})}(\mbf{\dot{a}}_{K}-\jj\pi (N-1)\cos(\vartheta_{K})\mbf{a}_{K})
    \end{array} \right] .\label{eq:T_theta}
    \end{equation}
    \hrule
\end{figure*}

Since we are primarily interested in the target angles $\vartheta_1,\,\ldots,\,\vartheta_K$, the CRLB-based objective function will be $\tr((\mbf{F}^{-1})_{K+1:2K,K+1:2K})$. By the Schur complement, we have
\begin{equation}
    (\mbf{F}^{-1})_{K+1:2K,K+1:2K} = \sigma^2_{\mrm{r}}\mbf{F}_{\vartheta}^{-1}
\end{equation}
where
\begin{equation}
    \mbf{F}_{\vartheta} = \mbf{T}_\vartheta^\HH\mbf{\tilde{X}}\mbf{T}_\vartheta - \mbf{T}_\vartheta^\HH\mbf{\tilde{X}}\mbf{T}_\alpha(\mbf{T}_\alpha^\HH\mbf{\tilde{X}}\mbf{T}_\alpha)^{-1}\mbf{T}_\alpha^\HH\mbf{\tilde{X}}\mbf{T}_\vartheta.
\end{equation}
Hence, the CRLB-based optimization metric is
\begin{equation}
    \fCRLB(\mbf{X}) = \sigma^2_{\mrm{r}}\tr\big(\mbf{F}_{\vartheta}^{-1}\big)\label{eq:fCRLB_obj}
\end{equation}
and the CRLB-based DFRC signal design problem is
\begin{equation}
\begin{aligned}
    &\minimize_{\{\mbf{X}\}}  &&  \sigma^2_{\mrm{r}}\tr\big(\mbf{F}_{\vartheta}^{-1}\big)\\
    & \st && \eqref{eq:beam_constrain1} \;\text{and}\; \eqref{eq:beam_constrain2}.
    \label{eq:problem3}
\end{aligned}
\end{equation}

Again, we can see that the optimization problem~\eqref{eq:problem3} is challenging to solve due to the complicated form of the objective function in~\eqref{eq:fCRLB_obj}, which involves nested matrix inversions. However, the SNC-based method proposed above can be adopted in this case to efficiently solve the signal design problem~\eqref{eq:problem3}. This is because the objective function $\fCRLB(\mbf{X})$ is (\textit{i}) non-convex w.r.t. $\mbf{X}$ with many local minima, (\textit{ii}) it is independent of the channel $\mbf{H}$ and the users' data $\mbf{s}$, and (\textit{iii}) it is invariant to any rotation $\bsm{\Omega}$ applied to the transmit signal $\mbf{X}$, i.e., $\fCRLB(\mbf{X}\bsm{\Omega}) = \fCRLB(\mbf{X})$, which can be easily verified by replacing $\mbf{X}$ by $\mbf{\tilde{X}} = \mbf{X}\bsm{\Omega}$ in the objective function in~\eqref{eq:fCRLB_obj}.

To construct the set of radar solutions for the CRLB-based metric, the challenge is in the derivation of the gradient of $\fCRLB(\mbf{X})$ w.r.t. $\mbf{X}$, which is addressed in Lemma~\ref{lemma2} below.
\begin{lemma}
    \label{lemma2}
    Let $\bsm{\Upsilon}_\vartheta =\mbf{T}_\vartheta\mbf{F}_{\vartheta}^{-2}\mbf{T}_\vartheta^\HH$ and $\bsm{\Upsilon}_\alpha =(\mbf{T}_\alpha^\HH\mbf{\tilde{X}}\mbf{T}_\alpha)^{-1}$. The Euclidean gradient of $\fCRLB(\mbf{X})$ w.r.t. $\mbf{X}$ is given by
    \begin{align}
        &\nabla \fCRLB(\mbf{X}) = \sigma_{\mrm{r}}^2 \mbf{\tilde{V}}^\TT\mbf{X}
        \label{eq:grad_fCRLB}
    \end{align}
    where the $(n,m)$-th element of the matrix $\mbf{\tilde{V}}$ is given as $[\mbf{\tilde{V}}]_{n,m} = \tr(\mbf{V}_{(n-1)N+1:nN,(m-1)N+1:mN})$ and:
    \begin{align}
    \mbf{V} &= -\bsm{\Upsilon}_\vartheta + \mbf{T}_\alpha\bsm{\Upsilon}_\alpha\mbf{T}_\alpha^\HH\mbf{\tilde{X}}\bsm{\Upsilon}_\vartheta \;-\notag \\
    &\quad\;\mbf{T}_\alpha \bsm{\Upsilon}_\alpha\mbf{T}_\alpha^\HH\mbf{\tilde{X}}\bsm{\Upsilon}_\vartheta\mbf{\tilde{X}}\mbf{T}_\alpha \bsm{\Upsilon}_\alpha\mbf{T}_\alpha^\HH +  \bsm{\Upsilon}_\vartheta\mbf{\tilde{X}}\mbf{T}_\alpha \bsm{\Upsilon}_\alpha\mbf{T}_\alpha^\HH.\label{eq:E_matrix}
    \end{align}
\end{lemma}
\begin{proof}
    See Appendix~\ref{appendix_2}.
\end{proof}

Using the gradient in~\eqref{eq:grad_fCRLB}, we can construct the set of radar solutions for the CRLB-based optimization problem following the same approach described in Section~\ref{sec:radar_solutions}. As in the cases involving the beampattern and SINR metrics, given a set of radar-only solutions for the CRLB optimization, we can use Algorithm~\ref{algo_R2DFRC} to obtain the final DFRC signal design for the CRLB-based metric.

\section{Accelerated SNC-Based Design}
\label{sec:aR2DFRC}

In this section, we propose an accelerated method referred to as aSNC that further reduces the computational complexity of the SNC approach. Unlike SNC, which performs $D$ alternating optimizations (AO) between $\varphi$ and $\bsm{\tilde{\phi}}$ to obtain $D$ DFRC candidate solutions $\{\mbf{x}^\RC_d\}$, the proposed aSNC-based method first finds the most likely radar solution and performs AO only for it.

Specifically, we first find $\Bar{\varphi}_d$ as the solution of~\eqref{eq:varphi_optimize} when $\bsm{\Tilde{\phi}}_d = \mbf{1}_N$ for all $d$.
Then, we divide the set $\mca{D} = \{1,\,\ldots,\,D\}$ into two disjoint subsets as follows:
\begin{align}
    &\mca{D}_1 = \big\{d \in \mca{D} \mid \min_{u}\, \big(\delta_u(\Bar{\varphi}_d) - \gamma_u\big) < 0 \big\},\\
    &\mca{D}_2 = \big\{d \in \mca{D} \mid \min_{u}\, \big(\delta_u(\Bar{\varphi}_d) - \gamma_u\big) \geq 0\big\}.
\end{align}
Hence, $\mca{D}_1$ is the set of indices $d$ whose solution $(\varphi_d, \bsm{\tilde{\phi}}_d) = (\bar{\varphi}_d, \mbf{1}_N)$ does not satisfy the communication constraints, while $\mca{D}_2$ is the set of indices $d$ whose solution $(\varphi_d, \bsm{\tilde{\phi}}_d) = (\bar{\varphi}_d, \mbf{1}_N)$ does. This means only the pairs $\{(\varphi_{d_1}, \bsm{\tilde{\phi}}_{d_1})\}$ with $d_1 \in \mca{D}_1$ need to be further updated by the AO strategy presented above. However, instead of performing AO on $\{(\varphi_{d_1}, \bsm{\tilde{\phi}}_{d_1})\}$ for all $d_1 \in \mca{D}_1$, we propose to do so for only the pair $\{(\varphi_{\hat{d}_1}, \bsm{\tilde{\phi}}_{\hat{d}_1})\}$ for which
\begin{equation}
\begin{aligned}
    \hat{d}_1 = \argmax_{\{d_1\in \mca{D}_1\}}  \;\;  \min_{u}\, \big(\delta_u(\Bar{\varphi}_{d_1}) - \gamma_u\big).
\end{aligned}
\label{eq:m1_hat}
\end{equation}

The intuition behind this choice is that making $\min_{u}\, \big(\delta_u(\Bar{\varphi}_{{d}_1}) - \gamma_u\big)$ large will lead to a pair $(\bar{\varphi}_{{d}_1}, \mbf{1}_N)$ that satisfies the communication constraints with a minimal number of updates on $\bsm{\tilde{\phi}}_{{d}_1}$, resulting in a minimal deviation from $\mbf{1}_N$, and consequently a minimal radar performance loss. This observation motivates us to perform AO for only the pair $\{(\varphi_{\hat{d}_1}, \bsm{\tilde{\phi}}_{\hat{d}_1})\}$ to obtain a DFRC solution
$$
\mbf{x}^\RC_{\hat{d}_1} = e^{\jj\varphi_{\hat{d}_1}}\diag(\mbf{x}^{\radar}_{\hat{d}_1})\bsm{{\tilde{\phi}}}_{\hat{d}_1}.
$$

Since the solutions $(\varphi_{d_2}, \bsm{\tilde{\phi}}_{d_2}) = (\bar{\varphi}_{d_2}, \mbf{1}_N)$ with $d_2 \in \mca{M}_2$ already satisfy the communication constraints, their DFRC solutions are given as
$$
\mbf{x}^\RC_{d_2} = e^{\jj\Bar{\varphi}_{d_2}}\diag(\mbf{x}^{\radar}_{d_2})\mbf{1}_N = e^{\jj\Bar{\varphi}_{d_2}}\mbf{x}^{\radar}_{d_2}.
$$
The final solution of the aSNC-based method is then obtained using
\begin{equation}
    \mbf{x}^\RC = \argmin_{\{\mbf{x}^\RC_{\hat{d}_1},\,\mbf{x}^\RC_{\hat{d}_2}\}}\; \left\{f_{\mca{A}}(\mbf{x}^\RC_{\hat{d}_1}),\,f_{\mca{A}}(\mbf{x}^\RC_{\hat{d}_2})\right\},
    \label{eq:x_aR2DFRC}
\end{equation}
where $\hat{d}_2$ is the index in $\mca{D}_2$ whose DFRC solution $\mbf{x}^\RC_{\hat{d}_2}$ gives the best radar performance among the candidates in $\mca{D}_2$, i.e.,
\begin{equation}
\begin{aligned}
    \hat{d}_2 = \argmin_{\{d_2\in \mca{D}_2\}}  \;  f_\mca{A}(\mbf{x}^\RC_{d_2}).
\end{aligned}
\label{eq:m2_hat}
\end{equation}

\begin{algorithm}[t!]
\caption{Proposed aSNC-Based Method.}
\label{algo_acce_R2JRC}
    \begin{algorithmic}[1]
        \REQUIRE $\mca{X}^{\radar}$, $\mbf{H}$, and $\mbf{s}$
        \ENSURE $\mbf{x}^{\RC}$
        \FOR{$d = 1, \ldots, D$}
            \STATE Compute $\mathbf{Q}_d = \diag(\mbf{s}^*)\mbf{H}\diag(\mbf{x}^{\radar}_d)$
            \STATE Set $\bsm{\Tilde{\phi}}_d = \mbf{1}_N$ and compute $\mbf{\tilde{s}}_d = \mbf{Q}_d\bsm{\tilde{\phi}}_d$
            \STATE Find $\varphi_d^\mrm{coarse}$ by~\eqref{eq:varphi_coarse} then set $\varphi_d^\mrm{fine} = \varphi_d^\mrm{coarse}$
            \STATE Update $\varphi_d^\mrm{fine}$ by~\eqref{eq:varphi_update} until convergence
            \STATE Set $\Bar{\varphi}_d = \varphi_d^\mrm{fine}$
        \ENDFOR
        \STATE Obtain $\hat{d}_1$ by~\eqref{eq:m1_hat}
        \WHILE{constraints $\delta_u \geq \gamma_u$ are not satisfied $\forall u$}
                \STATE Update $\bsm{\Tilde{\phi}}_{\hat{d}_1}$ by~\eqref{eq:phi_update}
    
                \STATE If $\delta_u \geq \gamma_u$ $\forall u$, then
                \textbf{exit} the while loop\label{algo:stop21}

                \STATE Compute $\mbf{\tilde{s}}_{\hat{d}_1} = \mbf{Q}_{\hat{d}_1}\bsm{\tilde{\phi}}_{\hat{d}_1}$
                \STATE Find $\varphi_{\hat{d}_1}^{\mrm{coarse}}$ by~\eqref{eq:varphi_coarse} then set $\varphi_{\hat{d}_1}^{\mrm{fine}} = \varphi_{\hat{d}_1}^{\mrm{coarse}}$
                \STATE Update $\varphi^\mrm{fine}_{d_1}$ by~\eqref{eq:varphi_update} until convergence
                \STATE Set $\varphi_{\hat{d}_1} = \varphi_{\hat{d}_1}^{\mrm{fine}}$

                \STATE If $\delta_u \geq \gamma_u$ $\forall u$, then
                \textbf{exit} the while loop\label{algo:stop22}
        \ENDWHILE
        \RETURN $\mbf{x}^\RC$ by~\eqref{eq:x_aR2DFRC}
    \end{algorithmic}
\end{algorithm}

\section{Computational Complexity Analysis}
\label{sec:complexity}
\begin{table}[t!]
\centering
\renewcommand{\arraystretch}{1.3}
	\centering
	\caption{Computational complexity comparison between the proposed methods and those in \cite{RLiu-JSTSP-2021, JYan-WCNC-2022, RLiu-JSAC-2022}, where $I$ is the number of iterations.\label{table:complexities}}
	{\small 
	\begin{tabular}{|l|c|}
		\hline
		\textbf{Algorithm}& \textbf{Complexity} \\
		\hline
		\textbf{PDD-MM-BCD}~\cite{RLiu-JSTSP-2021} & $\mathcal{O}\big(IN^3(U+N)\big)$ \\
		\hline
		\textbf{ALM-RBFGS}~\cite{RLiu-JSTSP-2021} & $\mathcal{O}\big(I(N^3+UN)\big)$ \\
		\hline
		\textbf{LADMM}~\cite{JYan-WCNC-2022} &  $\mathcal{O}\big( U N^3 + U N^2 + U^2 N)\big)$ \\
		\hline
        \textbf{MM-neADMM}~\cite{RLiu-JSAC-2022} &  $\mathcal{O}\big(I(N^3+UN)\big)$ \\
		\hline
        \textbf{Proposed SNC-based method} &  $\mathcal{O}\big(DI(C+N)U\big)$ \\
		\hline
  \textbf{Proposed aSNC-based method}&  $\mathcal{O}\big((D+I)(C+N)U\big)$ \\
		\hline
	\end{tabular}}
\end{table}

Here we analyze the computational complexity of the proposed approaches. We note that Algorithms \ref{algo_R2DFRC} and \ref{algo_acce_R2JRC} apply to all the considered radar metrics (beampattern similarity, SINR, and CRLB). Although these metrics result in different radar-only solution sets $\mathcal{X}^\radar$ to initialize these algorithms, constructing $\mathcal{X}^\radar$ can be performed offline, and the complexity of this process is excluded from the analysis. Thus, the following complexity analysis for Algorithms \ref{algo_R2DFRC} and \ref{algo_acce_R2JRC} applies to all considered radar metrics.

In Algorithm \ref{algo_R2DFRC}, the computation of $\mathbf{Q}_m$ in line~\ref{algo1_line3} requires $\mca{O}(UN)$ operations and needs to be performed only once for each radar solution. The computation of $\mbf{\tilde{s}}_m$ in line~\ref{algo1_line5}, $\varphi_m^{\mrm{coarse}}$ in line~\ref{algo1_line6}, and $\varphi_m^{\mrm{fine}}$ in line~\ref{algo1_line7} incur a computational complexity of $\mca{O}(UN)$, $\mca{O}(UC)$, and $\mca{O}(UI^{\mrm{fine}})$, respectively, where $I^{\mrm{fine}}$ represents the number of iterations for convergence in line~\ref{algo1_line7}. The computational complexity of updating $\bsm{\tilde{\phi}}_m$ in line~\ref{algo1_line10} is $\mca{O}(UN)$. It is typical that $I^{\mrm{fine}} < C$ since the update of $\varphi_m^{\mrm{fine}}$ in line~\ref{algo1_line7} is just a refining step. Consequently, the overall complexity of this algorithm is $\mathcal{O}\big(ID(C+N)U\big)$ where $I$ denotes the total number of iterations for the \textbf{while} loop. Building on this, the complexity of Algorithm~\ref{algo_acce_R2JRC} can be expressed as $\mathcal{O}\big((D+I)(C+N)U\big)$ where $\mathcal{O}\big(DU(C+N)\big)$ and $\mathcal{O}\big(IU(C+N)\big)$ are the computational complexity of the \textbf{for} and \textbf{while} loops in Algorithm~\ref{algo_acce_R2JRC}, respectively.



A complexity analysis comparing the proposed and existing algorithms is given in Table~\ref{table:complexities}. We see that the complexity of the proposed algorithms scales only linearly with the number of BS antennas $N$, a rate that is significantly lower than for the existing approaches, which scale at least as $N^3$. The complexity of the approaches in~\cite{RLiu-JSTSP-2021,JYan-WCNC-2022, RLiu-JSAC-2022} scale with $N^3$ since each of their iterations still requires either the inverse or eigenvalue decomposition of an $N\times N$ matrix.

\section{Numerical Results}
\label{sec_results}

\begin{figure}[t!]
    \centering
    \begin{subfigure}[t]{\linewidth}
        \centering
        \includegraphics[width=\linewidth]{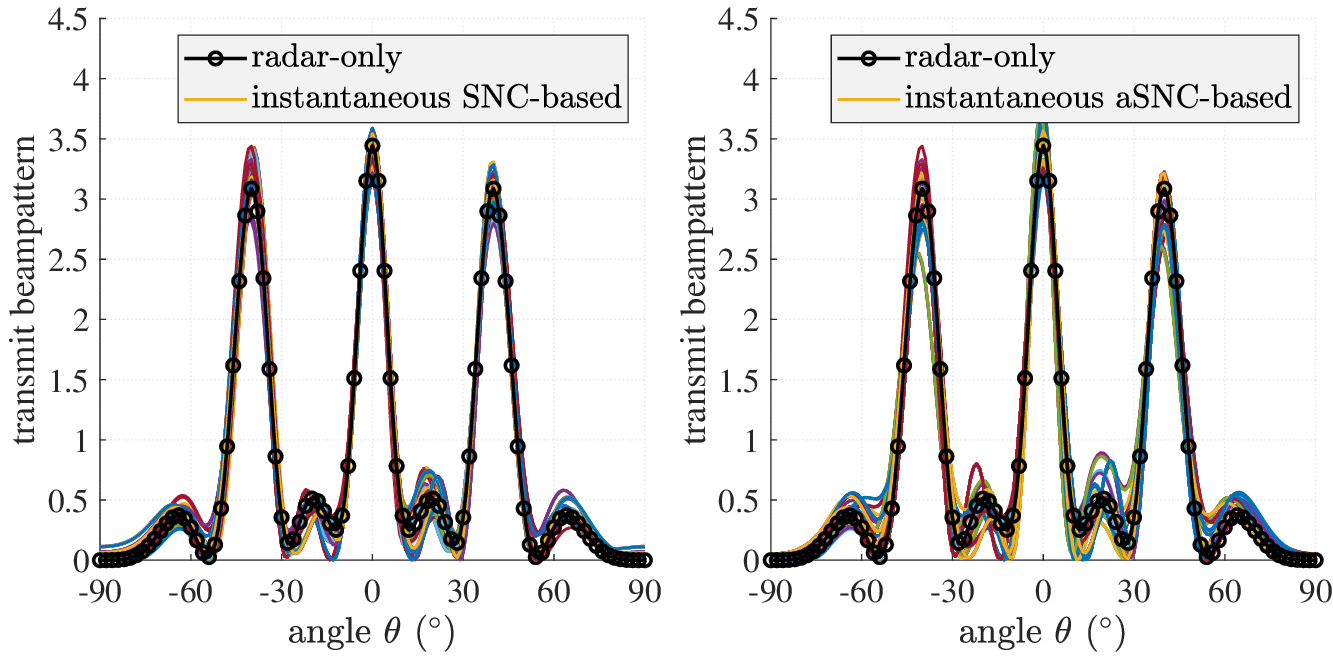}
        \caption{Transmit beampattern.}
        \label{fig_beampattern}
    \end{subfigure}
    
    \begin{subfigure}[t]{\linewidth}
        \centering
        \includegraphics[width=\linewidth]{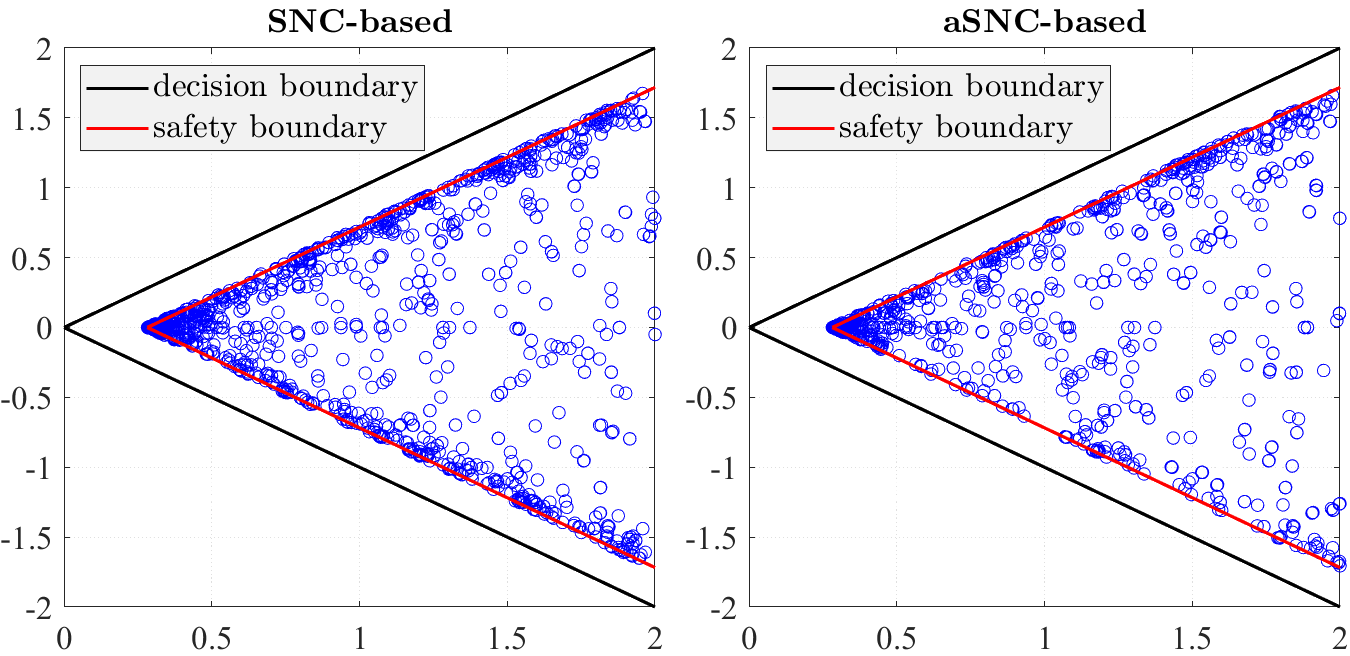}
        \caption{Rotated noiseless received signals}
        \label{fig_rx_signal}
    \end{subfigure}
    \belowcaptionskip = -10pt
    \caption{Transmit beampatterns and rotated noiseless received signals obtained by the proposed SNC and aSNC methods with $N = 10$, $U = 3$, $\Gamma = 9$ dB, and $P = 30$ dBm.}
    \label{fig_beam_vs_rx_signal}
\end{figure}

Here, we present numerical results to demonstrate the performance of the proposed algorithms and provide comparisons to previously proposed alternatives. We consider $K = 3$ targets, and unless otherwise specified, their angles are set to be $\{-40^\circ, 0^\circ, 40^\circ \}$. We also set $D=50$, $C = 100$, and $\{\theta_\ell\}$ to be the uniform angle samples in the beampattern between $-90^\circ$ and $90^\circ$ with a resolution of~$1^\circ$. The communication channels are modeled as $\mbf{h}_u\sim\mca{CN}(\mbf{0},\zeta_u\bsm{\Sigma}_u)$ where $\bsm{\Sigma}_u$ is the correlation matrix and $\zeta_u = \zeta_0 d_u^{-\nu}$ is the large-scale fading coefficient. We set $\zeta_0 = -30$ dB as the reference path loss and $\nu = 2.6$ as the path loss exponent. The distance $d_u$ between user $u$ and the BS is randomly generated between $100$ and $800$m and we set the distance beween the BS and the targets to be 3000m. We assume the noise power at the users is $-169$ dBm/Hz over a bandwidth of $1$ MHz. The transmit symbols are assumed to be 4-PSK, i.e., $M = 4$, and the safety margin thresholds $\gamma_u$ are set to be the same for all users as $\gamma_u = \gamma = \sigma_{\mrm{u}}\sin(\pi/M)\sqrt{\Gamma}$, where $\Gamma$ represents the SNR. We also set the step sizes $\eta_1 = 0.01$, $\eta_2 = 0.001$, and $\eta_3 = 0.005$.

\subsection{Beampattern}


\begin{figure}[t!]
    \centering
    \begin{subfigure}[t]{\linewidth}
        \centering
        \includegraphics[width=0.9\linewidth]{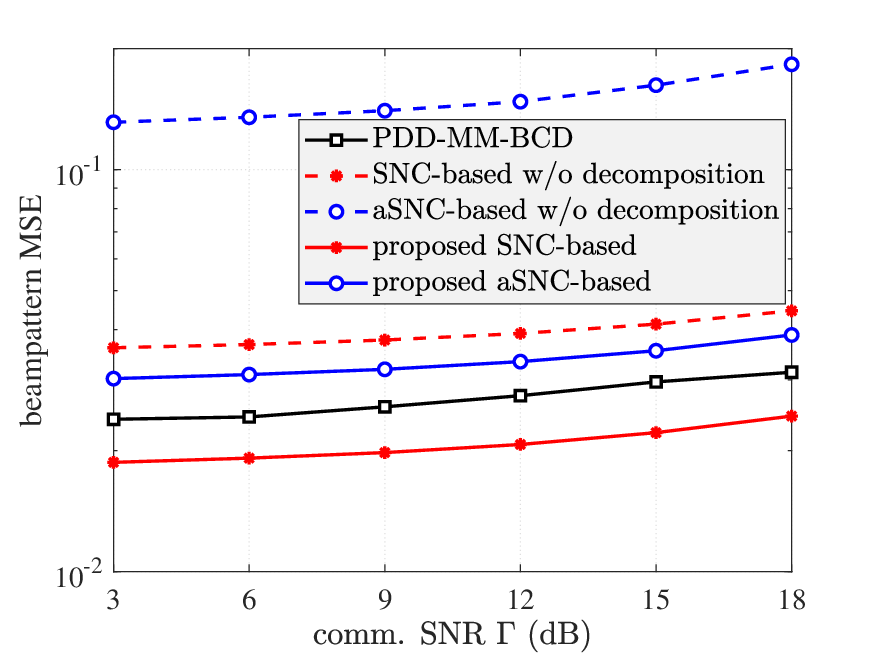}
        \caption{Beampattern MSE.}
        \label{fig_beam_MSE_vs_SNR}
    \end{subfigure}
    
    \begin{subfigure}[t]{\linewidth}
        \centering
        \includegraphics[width=0.9\linewidth]{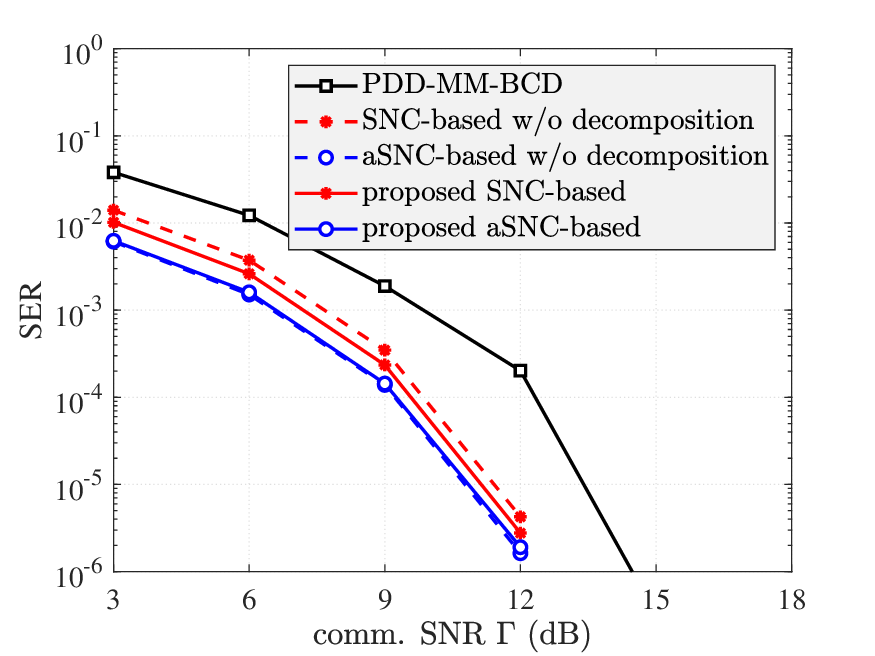}
        \caption{SER}
        \label{fig_beam_SER_vs_SNR}
    \end{subfigure}
    \belowcaptionskip = -15pt
    \caption{Comparison of beampattern MSE and communication SER for $N = 10$, $U = 3$, $P = 30$ dBm, and varying SNR constraint $\Gamma$.}
    \label{fig_beam_and_SER}
\end{figure}

First, we show the beampatterns and the rotated noiseless received signals of the proposed SNC-based and aSNC-based algorithms in Fig.~\ref{fig_beam_vs_rx_signal}. It can be seen that both algorithms create beampatterns that match the desired radar-only solution (Fig.~\ref{fig_beampattern}) while at the same time guaranteeing that the received signals at the users meet the required safety margin constraints. (Fig.~\ref{fig_rx_signal}). Fig.~\ref{fig_beampattern} also shows that, compared with aSNC, the beampatterns produced by SNC  align more closely with the desired radar-only result. This is reflected in Fig.~\ref{fig_beam_and_SER} which shows a comparison in terms of the communication symbol error rate (SER) and the beampattern mean squared error (MSE), which is defined as $\mbb{E}[\fbeam(\mbf{x})]$. Fig.~\ref{fig_beam_and_SER} also illustrates the performance of PDD-MM-BCD in~\cite{RLiu-JSTSP-2021} and a version of SNC and aSNC that does not use the decomposition in~\eqref{eq:rotation_decompose} but uses~\eqref{eq:phi_update} to directly find $\bsm{\phi}$. We use PDD-MM-BCD as the benchmark since it gives the best performance among other relevant approaches. Fig.~\ref{fig_beam_MSE_vs_SNR} shows that while the proposed SNC-based algorithm gives the best performance in terms of beampattern MSE, the versions SNC and aSNC that do not use the decomposition in~\eqref{eq:rotation_decompose} give the worst MSE. This verifies the effectiveness of~\eqref{eq:rotation_decompose}, which allows exploitation of the common phase $\varphi$. The MSE of the proposed aSNC algorithm is worse than SNC and PDD-MM-BCD since it only optimizes $\varphi$ and $\bsm{\tilde{\phi}}$ for the one most likely radar candidate. However, in terms of SER, the proposed aSNC approach offers the best performance while PDD-MM-BCD gives the highest SER. This is because aSNC creates larger safety margins while PDD-MM-BCD tends to make the users' safety margins equal the margin threshold, as can be seen in the cumulative distribution functions (cdf) in Fig.~\ref{fig:cdf}.

\begin{figure}[t!]
	\centering
	\includegraphics[width=0.9\linewidth]{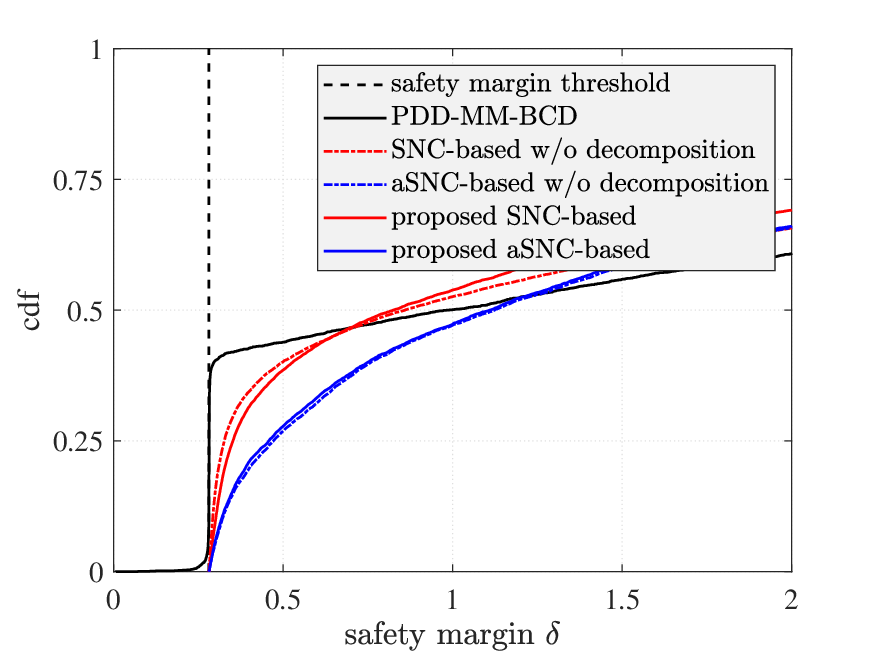}
    \belowcaptionskip = -10pt
	\caption{Empirical cdf of the safety margin for $N = 10$, $U = 3$, $\Gamma = 12$ dB and $P = 30$ dBm.}
	\label{fig:cdf}
\end{figure}

\subsection{SINR}
\begin{figure}[t!]
    \centering
    \begin{subfigure}[t]{0.9\linewidth}
        \centering
        \includegraphics[width=\linewidth]{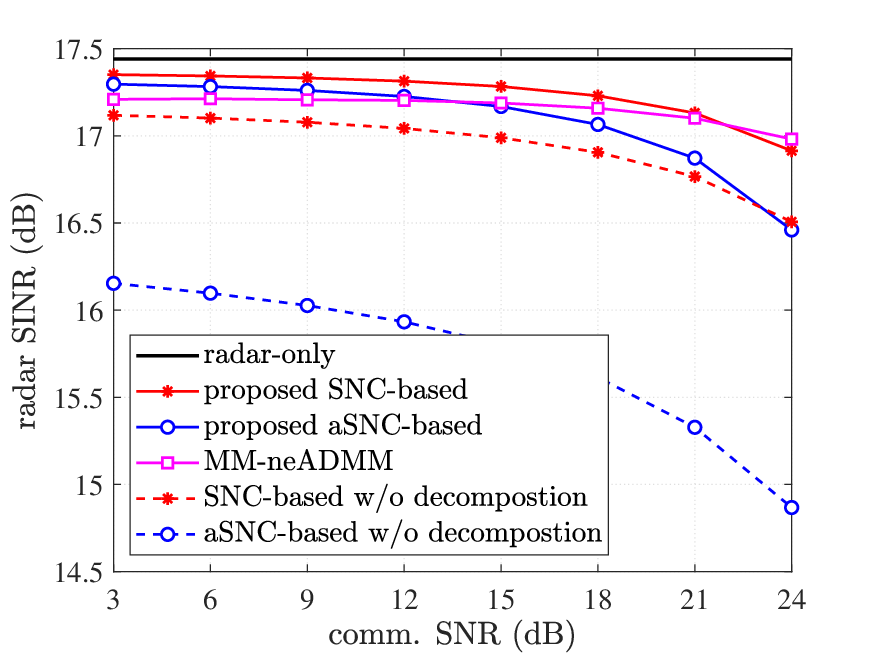}
        \caption{SINR performance comparison.}
        \label{fig_SINR_vs_SNR}
    \end{subfigure}
    
    \begin{subfigure}[t]{0.9\linewidth}
        \centering
        \includegraphics[width=\linewidth]{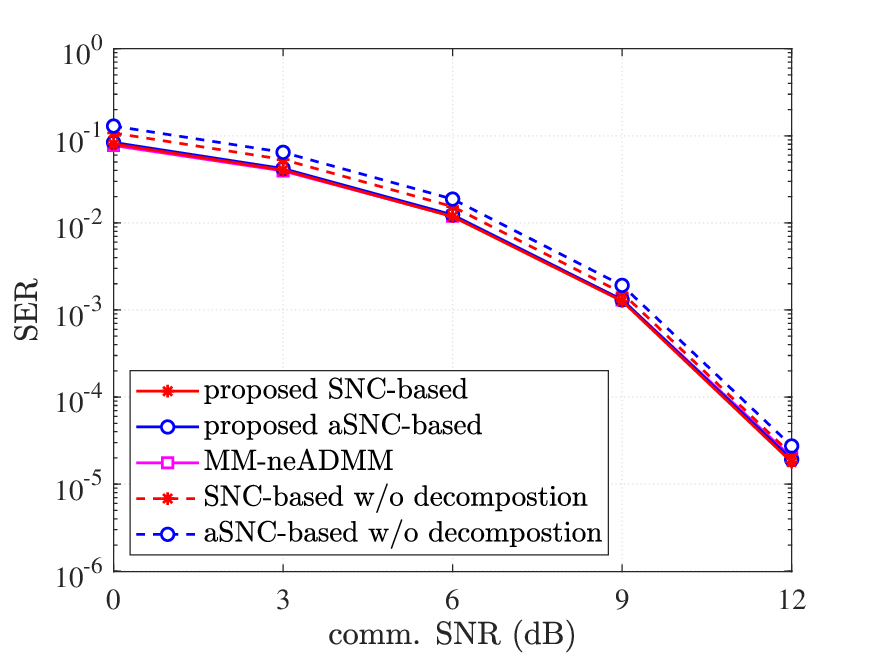}
        \caption{SER performance comparison}
        \label{fig_SER_vs_SNR}
    \end{subfigure}
    \belowcaptionskip = -10pt
    \caption{Comparison of radar SINR and communication SER for $N = 10$, $U = 3$, $P = 20$ dBm, and varying SNR constraint $\Gamma$.}
    \label{fig_SINR_vs_SER}
\end{figure}
\begin{figure}[t!]
    \centering
    \begin{subfigure}[t]{0.95\linewidth}
        \centering
        \includegraphics[width=\linewidth]{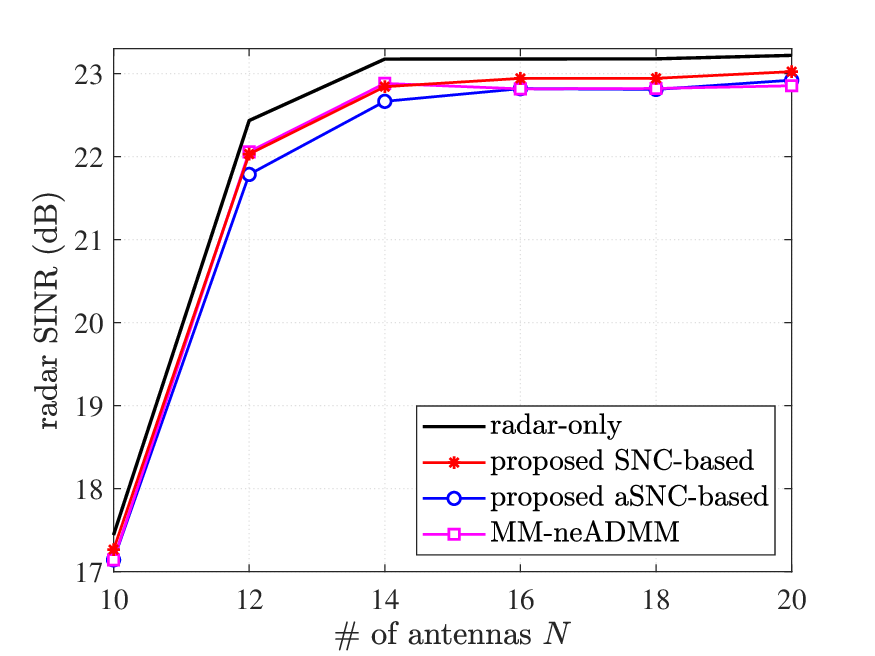}
        \caption{SINR versus number of antennas $N$ with $U = 3$.}
        \label{fig_SINR_vs_N}
    \end{subfigure}
    
    \begin{subfigure}[t]{0.95\linewidth}
        \centering
        \includegraphics[width=\linewidth]{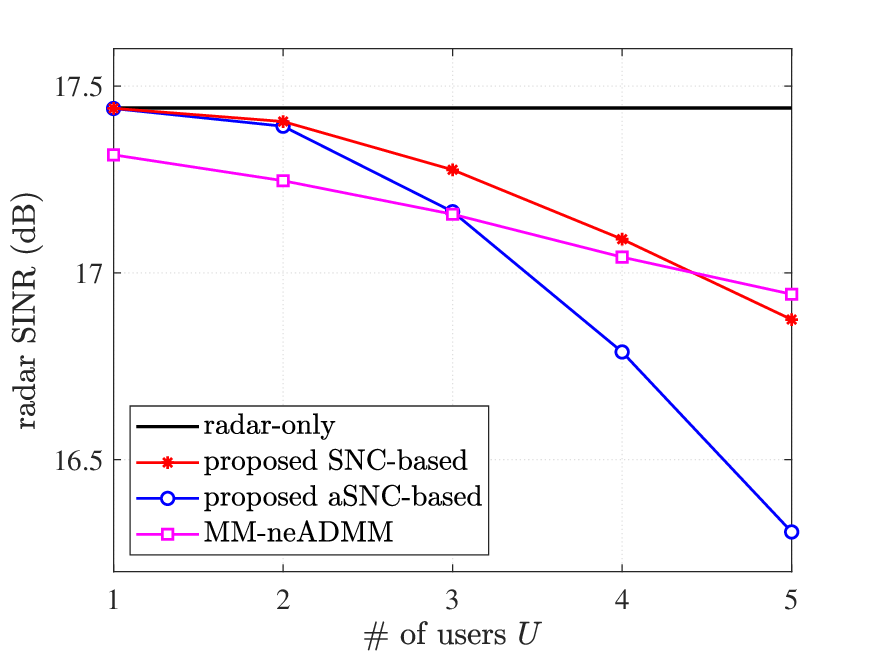}
        \caption{SINR versus number of users $U$ with $N = 10$.}
        \label{fig_SINR_vs_U}
    \end{subfigure}
    \belowcaptionskip = -15pt
    \caption{Comparison of SINR with varying numbers of antennas $N$ and users $U$ with communication constraint $\mrm{SNR} = 15$ dB and $P = 20$ dBm.}
    \label{fig_SINR_vs_U_N}
\end{figure}

Here, we evaluate performance using radar SINR. A radar SINR and communication SER comparison is given in Fig.~\ref{fig_SINR_vs_SER}, where we see that the proposed SNC algorithm gives the best radar SINR at low and medium communication SNR constraints, while the MM-neADMM method~\cite{RLiu-JSAC-2022} performs slightly better than SNC for high communication SNR constraints greater than 22 dB. For communication SNRs less than 12 dB, both SNC and aSNC  outperform MM-neADMM since a common phase rotation $\varphi$ by itself will be less likely to satisfy higher communication SNR constraints, resulting in greater use of the $\bsm{\tilde{\phi}}$ update in~\eqref{eq:phi_update} which sacrifices radar SINR to achieve the communication constraints. However, it should be noted that both SNC and aSNC have much lower computational complexity compared to MM-neADMM. The SER comparison in Fig.~\ref{fig_SER_vs_SNR} shows that SNC, aSNC, and MM-neADMM all give the same performance, indicating that the safety margin constraints are satisfied.

In Fig.~\ref{fig_SINR_vs_U_N}, we present the radar SINR for a fixed communication SNR constraint at 15 dB and different numbers of antennas $N$ (Fig.~\ref{fig_SINR_vs_N}) and users $U$ (Fig.~\ref{fig_SINR_vs_U}). The results in Fig.~\ref{fig_SINR_vs_N} show that the proposed SNC method achieves slightly higher SINR than aSNC and MM-neADMM, but all 3 methods provide performance within about 1/2 dB of the radar-only solution. While SINR performance increases significantly as $N$ increases to 12, it levels off for larger $N$ due to the limit on the transmit power. 

Fig.~\ref{fig_SINR_vs_U} shows a degradation in radar SINR as the number of users increases, since increasing $U$ means more communication constraints and subsequently fewer resources for sensing. When $U \le 4$, the proposed SNC-based method gives the best sensing performance while MM-neADMM gives the highest sensing SINR when $U> 5$. Again, this is because a large number of users means more communication SNR constraints, and a common phase rotation $\varphi$ alone will be less likely to help satisfy all of the SNR constraints. Thus more importance will inherently be placed on the $\bsm{\tilde{\phi}}$ update in~\eqref{eq:phi_update}, which sacrifices radar SINR for communication performance. It is interesting to note that when $U=1$, the radar SINR of the proposed SNC and aSNC methods is the same as that of the radar-only solution, i.e., there is no loss in radar sensing performance. In this case, the full flexibility of the common phase rotation $\varphi$ is able to rotate the user's received symbol so that its safety margin is satisfied without degrading the sensing performance.

\subsection{CRLB}
\begin{figure}[t!]
	\centering
	\includegraphics[width=0.95\linewidth]{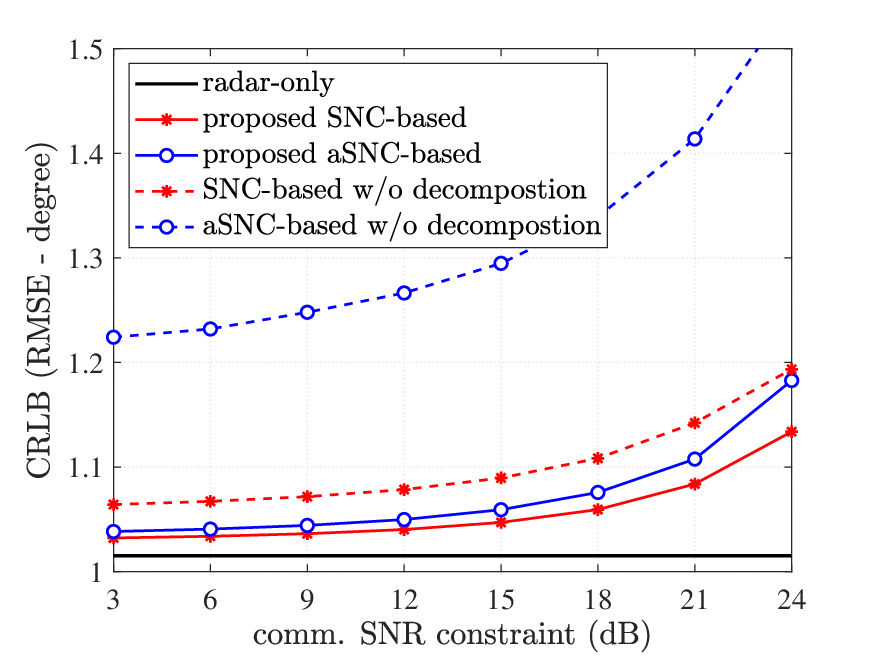}
	\caption{Comparison of CRLB for varying communication SNR constraint $\Gamma$, $N = 10$, $U = 3$, and $P = 20$ dBm.}
    \belowcaptionskip = -15pt
	\label{fig_CRLB_vs_SNR}
\end{figure}

In this section we evaluate the behavior of our proposed algorithms when the target angle CRLB is used as the performance metric. Fig.~\ref{fig_CRLB_vs_SNR} shows the CRLB for varying communication SNR constraints assuming three targets located at $\{-6^\circ, 0^\circ, 6^\circ\}$. For this scenario we observe that the proposed SNC method achieves the lowest CRLB, and as expected, the CRLB increases with stricter communication SNR constraints, reflecting the trade-off between communication and sensing. Fig.~\ref{fig_CRLB_vs_angleDiff} shows the improvement in the CRLB as the angular separation between the targets increases, where the communication SNR constraint is set to 15-dB. Note that in this case, both SNC and aSNC achieve performance essentially identical to that of the radar-only solution.

\begin{figure}[t!]
	\centering
	\includegraphics[width=0.95\linewidth]{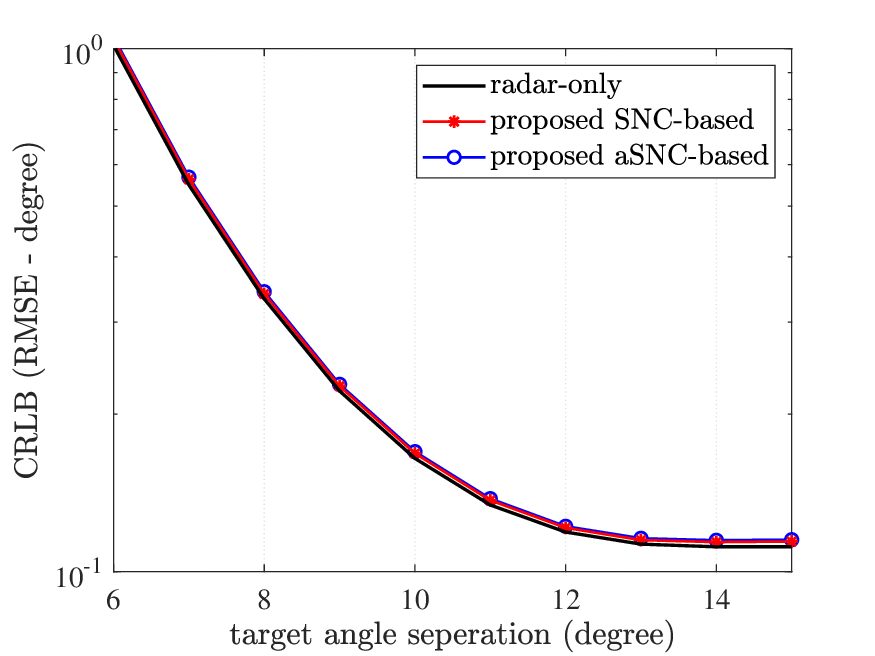}
    \belowcaptionskip = -15pt
	\caption{CRLB versus target angle separation with $N = 10$, $U = 3$, and $P = 20$ dBm.}
    \label{fig_CRLB_vs_angleDiff}
\end{figure}
\section{Conclusion}
\label{sec_conclusion}
In this paper we have developed two low-complexity yet highly effective algorithms for DFRC signal design using symbol level precoding. The proposed algorithms are based on the property of symmetric rotational invariance, which refers to the fact that a common symbol-rate phase modulation applied to the transmit signals across an antenna array does not affect the value of various radar sensing metrics, including beampattern similarity, SINR, and CRLB. We exploited this observation by perturbing the phase of the elements of a locally optimal radar-only waveform in order to (1) meet safety margin constraints for the communication users, and (2) yield a DFRC waveform that is as close as possible to the radar-only solution multiplied by a common phase modulation across all antennas. The proposed SNC- and aSNC-based algorithms were shown to be versatile and outperform existing methods in terms of radar sensing and communication performance as well as computational complexity.

\ifCLASSOPTIONcaptionsoff
  \newpage
\fi

\appendices
\section{Proof of Lemma~\ref{lemma1}.}
\label{appendix_1}
First, we compute the total derivative $df_{\SINR_k} (\mbf{x})$ w.r.t. $\mbf{x}$ and $\mbf{x}^\HH$, which is given as follows:
    \begin{align*}
        &df_{\SINR_k} (\mbf{x}) \\
        &= \sigma^2_{k}(d\mbf{x}^\HH)\mbf{B}_{k}^\HH\mbf{C}_k^{-1}\mbf{B}_{k}\mbf{x} + \sigma^2_{k}\mbf{x}^\HH\mbf{B}_{k}^\HH\mbf{C}_k^{-1}\mbf{B}_{k}d\mbf{x} \; + \\ &\hspace{4.5cm}\sigma^2_{k}\mbf{x}^\HH\mbf{B}_{k}^\HH(d\mbf{C}_k^{-1})\mbf{B}_{k}\mbf{x}\\
        &= \sigma^2_{k}(d\mbf{x}^\HH)\mbf{B}_{k}^\HH\mbf{C}_k^{-1}\mbf{B}_{k}\mbf{x} + \sigma^2_{k}\mbf{x}^\HH\mbf{B}_{k}^\HH\mbf{C}_k^{-1}\mbf{B}_{k}d\mbf{x} \; -  \\ &\hspace{3.5cm}\sigma^2_{k}\mbf{x}^\HH\mbf{B}_{k}^\HH\mbf{C}_k^{-1}(d\mbf{C}_k)\mbf{C}_k^{-1}\mbf{B}_{k}\mbf{x}\\
        &= \sigma^2_{k}(d\mbf{x}^\HH)\mbf{B}_{k}^\HH\mbf{C}_k^{-1}\mbf{B}_{k}\mbf{x} + \sigma^2_{k}\mbf{x}^\HH\mbf{B}_{k}^\HH\mbf{C}_k^{-1}\mbf{B}_{k}d\mbf{x} \;- \\
        &\;\;\;\;\,\sigma^2_{k}\mbf{x}^\HH\mbf{B}_{k}^\HH\mbf{C}_k^{-1}\times \\
        &\;\;\;\;\qquad\;\Big(\sum_{i\neq \ell_k} \sigma^2_i\mbf{B}_{i}\big((d\mbf{x})\mbf{x}^\HH + \mbf{x}d\mbf{x}^\HH\big)\mbf{B}_{i}^\HH\Big)\mbf{C}_k^{-1}\mbf{B}_{k}\mbf{x}\\
        &= (d\mbf{x}^\HH)\times \\
        &\;\;\;\;\;\Big[\sigma^2_{k}\big(\mbf{B}_{k}^\HH - \sum_{i\neq \ell_k}\sigma_i^2\mbf{x}^\HH\mbf{B}_{k}^\HH\mbf{C}_k^{-1}\mbf{B}_i\mbf{x}\mbf{B}_i^\HH\big)\mbf{C}_k^{-1}\mbf{B}_{k}\mbf{x}\Big] \\  
        &\;\;\;\;+\sigma^2_{k}\mbf{x}^\HH\mbf{B}_{k}^\HH\mbf{C}_k^{-1}\times \\
        &\hspace{2cm}\Big[\mbf{B}_{k} - \sum_{i\neq \ell_k}\sigma^2_i(\mbf{x}^\HH\mbf{B}_{i}^\HH\mbf{C}_k^{-1}\mbf{B}_{k}\mbf{x})\mbf{B}_i\Big]d\mbf{x}, \nbthis \label{eq:d_eta_k}
    \end{align*}
    where we have used the result $d\mbf{Z}^{-1} = -\mbf{Z}^{-1}(d\mbf{Z})\mbf{Z}^{-1}$~\cite[Eq.~3.40]{Jjorungnes2011Complex}.
    In addition, we have 
    \begin{equation}
        d f_{\SINR_k} (\mbf{x}) = \big(d\mbf{x}^\HH\big)\frac{\partial \eta_{k} (\mbf{x})}{\partial \mbf{x}^\HH} + \frac{\partial \eta_{k} (\mbf{x})}{\partial \mbf{x}} d\mbf{x}.
        \label{eq:d_eta_k_ref}
    \end{equation}
    Comparing~\eqref{eq:d_eta_k_ref} with~\eqref{eq:d_eta_k} yields
    \begin{align}
        &\frac{\partial f_{\SINR_k} (\mbf{x})}{\partial \mbf{x}^\HH} = \notag \\
        & \quad\sigma^2_{k}\Big(\mbf{B}_{k}^\HH - \sum_{i\neq \ell_k}\sigma^2_i(\mbf{x}^\HH\mbf{B}_{k}^\HH\mbf{C}_k^{-1}\mbf{B}_i\mbf{x})\mbf{B}_i^\HH\Big)\mbf{C}_k^{-1}\mbf{B}_{k}\mbf{x},
    \end{align}
    which is the Euclidean gradient of $f_{\SINR_k} (\mbf{x})$ in~\eqref{eq:grad_eta_k}.

\section{Proof of Lemma~\ref{lemma2}}
\label{appendix_2}
First we compute the total derivative of $\fCRLB$ w.r.t. $\mbf{X}$ and $\mbf{X}^*$, which is given as follows:
\begin{align*}
    &d \fCRLB \\
    &\;\;=\sigma^2_{\mrm{r}}\tr\big(\mbf{F}_{\vartheta}^{-1}\big) = -\sigma^2_{\mrm{r}}\tr\big(\mbf{F}_{\vartheta}^{-1}d(\mbf{F}_{\vartheta})\mbf{F}_{\vartheta}^{-1}\big)\\
    &\;\;= -\sigma^2_{\mrm{r}}\tr\big(\mbf{F}_{\vartheta}^{-2}\mbf{T}_\vartheta^\HH d(\mbf{\tilde{X}})\mbf{T}_\vartheta\big) \; + \\
    &\;\;\quad\; \sigma^2_{\mrm{r}}\tr\big(\mbf{F}_{\vartheta}^{-2}\mbf{T}_\vartheta^\HH d(\mbf{\tilde{X}})\mbf{T}_\alpha(\mbf{T}_\alpha^\HH\mbf{\tilde{X}}\mbf{T}_\alpha)^{-1}\mbf{T}_\alpha^\HH\mbf{\tilde{X}}\mbf{T}_\vartheta \big) \; + \\
    &\;\;\quad\; \sigma^2_{\mrm{r}}\tr\big(\mbf{F}_{\vartheta}^{-2}\mbf{T}_\vartheta^\HH\mbf{\tilde{X}}\mbf{T}_\alpha d((\mbf{T}_\alpha^\HH\mbf{\tilde{X}}\mbf{T}_\alpha)^{-1})\mbf{T}_\alpha^\HH\mbf{\tilde{X}}\mbf{T}_\vartheta \big) \; + \\
    &\;\;\quad\; \sigma^2_{\mrm{r}}\tr\big(\mbf{F}_{\vartheta}^{-2}\mbf{T}_\vartheta^\HH\mbf{\tilde{X}}\mbf{T}_\alpha (\mbf{T}_\alpha^\HH\mbf{\tilde{X}}\mbf{T}_\alpha)^{-1}\mbf{T}_\alpha^\HH d(\mbf{\tilde{X}})\mbf{T}_\vartheta \big)\\
    &\;\;=\sigma^2_{\mrm{r}}\tr (\mbf{V}d(\mbf{\tilde{X}})), \nbthis \label{eq:df_CRLB_0}
\end{align*}
where the matrix $\mbf{V}$ is given in~\eqref{eq:E_matrix}. The term $\tr (\mbf{V}d(\mbf{\tilde{X}}))$ in~\eqref{eq:df_CRLB_0} can be written in the following form:
\begin{align*}
    &\tr (\mbf{V}d(\mbf{\tilde{X}})) \notag \\
    &\;= \tr\big(\mbf{V} (d(\mbf{X}^*)\mbf{X}^\TT\otimes\mbf{I}_N) \big) + \tr\big(\mbf{V} (\mbf{X}^*d(\mbf{X})^\TT\otimes\mbf{I}_N) \big)\\
    &\;= \tr\big(\mbf{\tilde{V}} (d(\mbf{X}^*)\mbf{X}^\TT) \big) + \tr\big(\mbf{\tilde{V}} (\mbf{X}^*d(\mbf{X})^\TT) \big)\\
    &\;= \tr\big(\mbf{X}^\TT\mbf{\tilde{V}}d(\mbf{X}^*)\big) + \tr\big(d(\mbf{X})\mbf{X}^\HH\mbf{\tilde{V}}^\TT \big).
\end{align*}
Therefore, the total derivative of $\fCRLB$ w.r.t. $\mbf{X}$ and $\mbf{X}^*$ is
\begin{align}
    d \fCRLB = \tr\big(\sigma^2_{\mrm{r}}\mbf{X}^\TT\mbf{\tilde{V}}d(\mbf{X}^*)\big) + \tr\big(\sigma^2_{\mrm{r}}d(\mbf{X})\mbf{X}^\HH\mbf{\tilde{V}}^\TT \big).\label{eq:df_CRLB}
\end{align}
Comparing the total derivative in~\eqref{eq:df_CRLB} with the fact that
\begin{equation}
    d \fCRLB = \tr\left(\frac{\partial \fCRLB}{\partial \mbf{X}^*}d \mbf{X}^*\right) + \tr\left(d\mbf{X}\frac{\partial \fCRLB}{\partial \mbf{X}} \right)
\end{equation}
implies
\begin{equation}
    \frac{\partial \fCRLB}{\partial \mbf{X}^*} = \sigma_{\mrm{r}}^2 \mbf{X}^\TT\mbf{\tilde{V}}.
\end{equation}
Therefore, the Euclidean gradient of $\fCRLB(\mbf{X})$ w.r.t. $\mbf{X}$ is
\begin{equation}
    \nabla_{\mbf{X}} \fCRLB = \frac{\partial \fCRLB}{\partial \mbf{X}^\HH} = \sigma_{\mrm{r}}^2 \mbf{\tilde{V}}^\TT\mbf{X},
\end{equation}
which is the result in~\eqref{eq:grad_fCRLB}.

\bibliographystyle{IEEEtran}
\bibliography{ref}

%









\end{document}